\documentclass[10pt, conference, compsocconf]{IEEEtran}
\IEEEoverridecommandlockouts

\usepackage{cite}  
\usepackage{amsmath,amssymb,amsfonts,amsthm}
\usepackage{mathrsfs}
\usepackage{algorithm,algpseudocode}
\usepackage{graphicx}
\usepackage{textcomp}
\usepackage{xcolor}
\usepackage{cancel}
\usepackage[normalem]{ulem}
\usepackage{url}
\usepackage[bookmarks=false]{hyperref}
\usepackage[subrefformat=parens,labelformat=parens]{subfig}
\usepackage{multirow}
\usepackage{verbatim}
\usepackage{array}
\usepackage{balance}
\usepackage{booktabs}
\usepackage{paralist}
\usepackage{tabularx}
\usepackage{dblfloatfix}
\usepackage{fixltx2e}
\usepackage{color}
\usepackage{textcomp}
\usepackage{svg}
\usepackage{lipsum,multicol}
\usepackage{bm}  
\usepackage{upgreek}
\usepackage{centernot}
\usepackage{todonotes}
\usepackage{marginnote}
\usepackage{caption}
\usepackage{fancyhdr}

\fancypagestyle{firstpage}{
    \fancyhf{} 
    \fancyhead[L]{Revision 1} 
    \fancyhead[R]{August 2024} 
    \fancyfoot[L]{Revision: This document has been updated to correct Lemma 1 and Theorems 5 and 6. We thank Dr. Jian-Jia Chen of Technische Universität Dortmund for alerting us to the issue.}

}

\def\BibTeX{{\rm B\kern-.05em{\sc i\kern-.025em b}\kern-.08em
    T\kern-.1667em\lower.7ex\hbox{E}\kern-.125emX}}

\renewcommand\footnoterule{\kern-3pt \hrule width 2in \kern 2.6pt}

\newcolumntype{L}[1]{>{\raggedright\let\newline\\\arraybackslash\hspace{0pt}}m{#1}}
\newcolumntype{C}[1]{>{\centering\let\newline\\\arraybackslash\hspace{0pt}}m{#1}}
\newcolumntype{R}[1]{>{\raggedleft\let\newline\\\arraybackslash\hspace{0pt}}m{#1}}

\algtext*{EndWhile}
\algtext*{EndIf}
\algtext*{EndFor}

\newtheorem{theorem}{Theorem}

\newtheorem{lemma}{Lemma}
\newtheorem{definition}{Def.}

\newcommand{\blue}[1]{{{\color{blue} #1}}}

\begin{document}
\pagestyle{plain}
\pagenumbering{arabic}

\title{
Timing Analysis and Priority-driven Enhancements of ROS 2 Multi-threaded Executors
}

\author{
	\IEEEauthorblockN{Hoora Sobhani$^{\star}$\qquad Hyunjong Choi$^{\dagger}$\qquad Hyoseung Kim$^{\star}$}
	\IEEEauthorblockA{$^{\star}$University of California, Riverside\\$^{\dagger}$San Diego State University}
	\IEEEauthorblockA{hsobh002@ucr.edu, hyunjong.choi@sdsu.edu, hyoseung@ucr.edu}
}

\maketitle

\thispagestyle{firstpage}   

\begin{abstract}
The second generation of Robotic Operating System, ROS 2, has gained much attention for its potential to be used for safety-critical robotic applications. The need to provide a solid foundation for timing correctness and scheduling mechanisms is therefore growing rapidly. 
Although there are some pioneering studies conducted on formally analyzing the response time of processing chains in ROS 2, the focus has been limited to single-threaded executors, and multi-threaded executors, despite their advantages, have not been studied well. 
To fill this knowledge gap, in this paper, we propose a comprehensive response-time analysis framework for chains running on ROS 2 multi-threaded executors. We first analyze the timing behavior of the default scheduling scheme in ROS 2 multi-threaded executors, and then present priority-driven scheduling enhancements to address the limitations of the default scheme. Our framework can analyze chains with both arbitrary and constrained deadlines and also the effect of mutually-exclusive callback groups. 
Evaluation is conducted by a case study on  NVIDIA Jetson AGX Xavier and schedulability experiments using randomly-generated chains. The results demonstrate that our analysis framework can safely upper-bound response times under various conditions and the priority-driven scheduling enhancements not only reduce the response time of critical chains but also improve analytical bounds. 
\end{abstract}


\section{Introduction}
The Robotic Operating System (ROS) is an open-source middleware framework that has been widely used for robotic systems in academia and industry. The software modularity and composability of ROS have helped the community achieve efficient and productive robotic software developments. However, the architecture limitations and several deep-rooted shortcomings of ROS had been unveiled over the decades, resulting in the development of its second generation, ROS 2, which is a complete refactoring of the previous version.

One of the major considerations in ROS 2 has been improving real-time capabilities while inheriting the successful concepts of its predecessor. As an example, to support real-time data distribution, ROS 2 employs the Data Distribution Service (DDS) as the underlying communication framework. Although ROS 2 has been shown to provide better real-time support for robotic systems, it is yet incomplete to be applicable to hard real-time or safety-critical applications. To guarantee stringent timing constraints in these applications, designers need to safely upper-bound the end-to-end latency (i.e., response time) of \textit{processing chains}. Although there are many prior studies on the response-time analysis of chains, the unique scheduling behavior of ROS 2 calls for new formal modeling and analysis of its timing abstractions and scheduling architecture.

The pioneers in formally analyzing the response time of chains on ROS 2 are \cite{Casini_ECRTS19, Tang_RTSS20, PICAS}. As mentioned in these studies, ROS 2 introduces ``executors'' as the abstraction of operating system (OS) processes, providing two built-in types: \textit{single-threaded} and \textit{multi-threaded}. A single-threaded executor executes callbacks sequentially, while a multi-threaded executor distributes pending callbacks across multiple threads (i.e., callbacks can execute in parallel). These studies focus on the response-time analysis of callbacks and chains {\it only} on single-threaded executors. 
Specifically, \cite{Casini_ECRTS19} and \cite{Tang_RTSS20} mapped a single-threaded executor to a single reservation server to derive analysis; \cite{PICAS} proposed priority-driven scheduling and executor-to-core allocation but for single-threaded executors. 

As of yet, the scheduling behavior of ROS 2 multi-threaded executors has not been studied well. However, plenty of studies in the real-time systems area have demonstrated that multi-threading improves system concurrency and throughput by effectively utilizing multiple processors while preserving timing correctness, e.g.,  real-time multi-threading in self-driving cars~\cite{KIM_ICCPS13}. Therefore, in this paper, we aim to analyze and improve the timing behavior of ROS 2 multi-threaded executors. But, the tremendous amount of non-determinism in multi-threaded executors, such as an unpredictable distribution of callbacks across threads and unsynchronized polling points of threads, makes the analysis particularly challenging. In addition, the lack of systematic support for chain priority in ROS 2 prevents the effective utilization of parallel resources, resulting in delayed processing of critical chains.

This paper tackles the aforementioned issues. We first present a response-time analysis (RTA) framework for chains running on ROS 2 multi-threaded executors. To improve the end-to-end response time of critical chains, we also propose priority-driven scheduling enhancements that make the executor strictly respect the priority of the corresponding chain when scheduling individual callbacks. These enhancements bring significant benefits in timing analysis as well as observed performance on a real platform. The detailed contributions of our work are as follows:



\begin{itemize}

\item We discuss difficulties in analyzing the timing behavior of chains on multi-threaded executors (Sec.~\ref{challenges}). In particular, we redefine the properties of two ROS-specific scheduling behaviors, polling points and processing windows, and explain why the latest single-threaded analysis based on them is not applicable to multi-threaded executors.

\item We develop an RTA framework for ROS 2 multi-threaded executors (Sec.~\ref{RTA}). 
Our analysis considers chains with both constrained and arbitrary deadlines, and upper-bounds the response time of chains executed by multi-thread executors. 
We also analyze the effects of {\em callback groups} that are used to control the concurrency of select callbacks.

\item We propose priority-driven scheduling enhancements and the corresponding analytical extensions to our RTA framework (Sec. \ref{RTA}). The priority-driven scheduling approach mitigates the aforementioned non-determinism issues and the resulting analytical pessimism and helps reduce the response time of critical chains.


\item For evaluation, we performed a case study movitated by autonomous driving software on an embedded platform as well as schedulability experiments using randomly-generated workloads (Sec. \ref{EVAL}). The results support the effectiveness of our RTA framework and demonstrate how priority-driven scheduling enhancements improve both observed and computer upper-bounds on the response time of chains
\end{itemize}

\section{Related Work}
Many studies have been conducted on improving real-time capabilities~\cite{saito2018rosch, wei2016rt} and evaluating the empirical real-time performance of ROS~\cite{gutierrez2018towards, maruyama2016exploring}. In~\cite{wei2016rt}, Wei et al. proposed to run two OSes on the same platform, i.e., real-time ROS nodes on Nuttx (RTOS) and non-real-time ones on Linux, to provide isolated execution environments. Saito et al.~\cite{saito2018rosch} developed ROSCH-G, which is a real-time extension to ROS with a CPU/GPU coordination mechanism provided as a loadable kernel module.
Carlos et al.~\cite{gutierrez2018towards} measured the worst-case latency between two nodes and observed deadline miss behavior in a Linux system with the PREEMPT-RT patch. In~\cite{maruyama2016exploring}, the effects of various QoS configurations under different vendor-specific DDS implementations were evaluated empirically. However, some of these studies were conducted on the first generation of ROS~\cite{saito2018rosch, wei2016rt} and the others did not consider formal modeling or analysis of ROS 2~\cite{gutierrez2018towards, maruyama2016exploring}.

Formal timing analysis of end-to-end latency has recently received much attention for processing chains that follow either the publisher-subscriber or read-execute-write model. Davare et al.~\cite{davare2007period} and Schlatow et al.~\cite{schlatow2016response} captured an upper bound on the end-to-end latency of a chain based on the worst-case response time of individual tasks. In~\cite{abdullah2019worst,becker2016synthesizing,kloda2018latency}, the authors proposed analytical methods to bound the end-to-end latency of a chain for fixed-priority scheduling. Choi et al.~\cite{choi20chain} focused on improving the end-to-end latency of chains and proposed chain-based fixed-priority scheduling. However, these approaches cannot be directly applied to ROS 2 due to discrepancies in the scheduling model.

The literature on the response-time analysis of processing chains on ROS 2 executors is quite limited. The first work that formally analyzed the timing behavior of ROS 2 executors is the work by  Casini et al.~\cite{Casini_ECRTS19}. They unveiled the details of the ROS 2 callback scheduling policy implemented in single-threaded executors and presented the response-time analysis of callbacks and processing chains on single-threaded executors.
Tang et al.~\cite{Tang_RTSS20} proposed an improved response-time analysis by characterizing the details of processing windows (we will revisit this in Sec.~\ref{challenges}), and studied the effect of callback priorities under the standard ROS 2 scheduling policy that results in round-robin-like behavior~\cite{tang2021response}.
In \cite{PICAS}, Choi et al. developed a priority-driven chain-aware scheduler, called PiCAS, which modifies the default ROS 2 policy to strictly follow assigned chain priorities. They also presented callback priority assignment, allocation of nodes to executors and executors to CPU cores, and response-time analysis under their scheduler. Unlike PiCAS which uses fixed-priority scheduling, Arafat et al.~\cite{arafat2022response} proposed a dynamic-priority scheduling scheme to improve chain latency, especially in an overloaded scenario. Teper at el.~\cite{teper2022end} focused on cause-effect chains where intermediate callbacks can be released independently by their own timers, and proposed an end-to-end latency analysis of cause-effect chains in ROS 2. However, all of these focus on single-threaded executors, none on multi-threaded executors.

Both single-threaded and multi-threaded ROS 2 executors have yet another important feature called {\em callback groups} that have not been considered in the prior analysis work~\cite{Casini_ECRTS19,Tang_RTSS20,PICAS, rage, arafat2022response, teper2022end}. While the authors of \cite{yang20exploring} explored the effect of callback groups in single-threaded Micro-ROS executors (a variant of ROS 2 for microcontrollers), they did not link this to formal timing analysis. In this paper, we take this into account in our RTA framework. 

Recently, Jiang et al. \cite{jiang2022real} presented response-time analysis for processing chains with constrained deadlines running on a ROS 2 multi-threaded executor. While the reader might find it similar to our work, we make unique contributions in that our work analyzes chains with both constrained and arbitrary deadlines and provides priority-driven scheduling enhancements and the corresponding analytical extensions to ROS 2 multi-threaded executors.

\section{Background and System Model}
In this section, we briefly review the ROS 2 architecture and introduce our system model.

\begin{figure}[t]
\centerline{\includegraphics[width=1\linewidth]{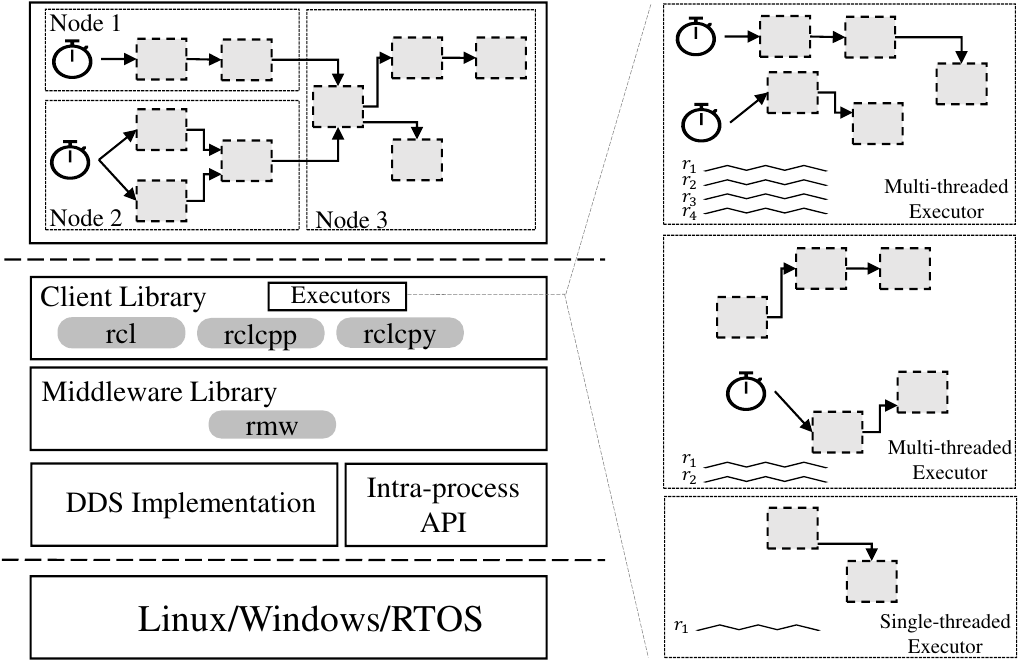}}
\caption{ROS 2 architecture and application model}
\label{fig1}
\end{figure}

\subsection{ROS 2 Architecture}\label{sec:architecture}

Fig.~\ref{fig1} illustrates the ROS 2 middleware architecture that sits on top of the OS and provides a set of libraries and tools for robotic application development. In particular, the ROS~2 client library (rcl) consists of language-specific libraries (rclcpp and rclpy for C++ and Python, respectively), and the middleware library (rmw) provides publisher-subscriber interfaces for Data Distribution Service (DDS) and intra-process communication.

ROS 2 applications consist of \textit{nodes}, each of which in turn consists of a set of \textit{callbacks}. ROS 2 callbacks are categorized into four types: timer callbacks are triggered when the timer period is up;  subscription callbacks are triggered when a subscribed message arrives; service and client callbacks are triggered by a service request and a response to a service request, respectively. There is an inherent priority order among these callback types such that timer callbacks have the highest and client callbacks have the lowest priority~\cite{Casini_ECRTS19,Tang_RTSS20,PICAS, rage}. For callbacks of the same type, their priorities are implicitly determined by declaration order in the node.
While callbacks are actual executable entities, nodes are just containers of callbacks and serve as an abstraction to allocate callbacks to executors, i.e., once a node is assigned to an executor, all callbacks of that node are executed by that executor.

A set of callbacks with data dependencies forms a \textit{processing chain}, or simply called a chain. Each callback can be associated with one or more chains, and chains can be formed by callbacks from different nodes. We will give a more detailed explanation of chains in Sec.~\ref{sec:system_model}.

An \textit{executor} is the ROS 2 abstraction of OS-level scheduling entities and it executes callbacks assigned to it. Each executor maintains a \textit{ReadySet}, which is a cached set of ``ready'' regular (non-timer) callbacks~\cite{Casini_ECRTS19,Tang_RTSS20}.\footnote{Timer callbacks are added to the ReadySet instantly after their release.} The ReadySet is updated only when it is empty (or there is no callback in the set eligible to execute\footnote{Ready callbacks in the ReadySet may not be eligible to run due to callback groups that enforce concurrency. We will explain more details later.}). This is the only point when the executor communicates with underlying layers, and is called a \textit{polling point} in the literature. The time interval between two consecutive polling points is called a \textit{processing window}, in which the executor processes regular callbacks in its ReadySet plus incoming timer callbacks. Therefore, callbacks released after one polling point are not processed until the next polling point, causing priority inversion. Also, the use of the ReadySet makes at most one instance of any callback be processed in one processing window. Hence, in addition to the priority inversion problem, a callback instance might be blocked for multiple processing windows before it gets scheduled if there are multiple pending instances of the same callback. Sec.~\ref{challenges} discusses more details on these issues in the context of multi-threaded executors.

In addition, ROS 2 executors provide \textit{callback groups} to control the concurrency of callback execution. There are two options: \textit{reentrant} and \textit{mutually-exclusive}. The reentrant callback group allows an executor to execute any ready callback with no other restriction. Hence, as long as a preceding callback in a chain has completed execution, the next callback becomes ready and can be considered for scheduling. On the other hand, the mutually-exclusive callback group limits any callback within this group not to be executed in parallel. In other words, the execution of callbacks within a mutually-exclusive group is all serialized even if the executor has multiple threads. This could be a useful option if callbacks were originally programmed for a single-threaded environment but assigned to a multi-threaded executor, or they access shared data with no synchronization in mind. However, the use of mutually-exclusive groups introduces another type of dependency in callback scheduling. Our analysis in Sec.~\ref{RTA} takes into account the effects of these two options.

\subsection{System Model}\label{sec:system_model}
We consider a ROS 2 system $\Gamma$ running on a multi-core platform. The system $\Gamma$ is composed of a set of independent chains, i.e.,
$\Gamma = \{\Gamma_\mathcal{C}, \Gamma_{\mathcal{C}'}, \Gamma_{\mathcal{C}''}, ...\}$. We present our model for callbacks, chains, and executors as below. Without loss of generality, we assume the discrete-time model in our system, in which a time interval is a non-negative integer multiplier of the system time unit (e.g., clock cycle).



    



\smallskip
\noindent\textbf{Callbacks.}
A callback $\tau_{ \langle \mathcal{C},j \rangle}$ is denoted as the $j^{th}$ callback of a chain $\Gamma_{\mathcal{C}}$, where $1 \leq j \leq \|\Gamma_\mathcal{C}\| $. We characterize each callback $\tau_{ \langle \mathcal{C},j \rangle}$ with:
\begin{itemize}
    \item \textbf{$E_{\langle \mathcal{C},j \rangle}$}: The worst-case execution time (WCET) of an instance of a callback $\tau_{ \langle \mathcal{C},j \rangle}$.
    


    \item \textbf{$\pi_{\langle \mathcal{C},j \rangle}$}: The priority of $\tau_{ \langle \mathcal{C},j \rangle}$ (smaller values mean lower priority).
\end{itemize}
Note that callbacks belonging to a chain do not have their own periods or deadlines since they follow their chain's timing constraints. If the system has a callback that does not belong to any chain but has a timing constraint, it can be modeled as a single-callback chain for analysis purposes.

\smallskip
\noindent\textbf{Chains.} A chain $\Gamma_\mathcal{C} = \{\tau_{\langle \mathcal{C},1 \rangle}, \tau_{ \langle \mathcal{C},2 \rangle}, ..., \tau_{ \langle \mathcal{C}, \|\Gamma_\mathcal{C}\| \rangle}\}$ consists of $\|\Gamma_\mathcal{C}\|$ callbacks with a sequential execution order, i.e., $\tau_{ \langle \mathcal{C},j+1 \rangle}$ can start only after $\tau_{ \langle \mathcal{C},j \rangle}$ finishes. 
The chain $\Gamma_\mathcal{C}$ can be either timer-triggered (i.e., the first callback $\tau_{\langle \mathcal{C},1 \rangle}$ is a timer callback) or event-triggered (i.e., $\tau_{\langle \mathcal{C},1 \rangle}$ is a regular callback released by a periodic sensor input or message from other sub-systems).
Due to data dependency, any subsequent callback becomes ready to run only when its predecessor finishes execution.
For presentation simplicity, the periods of all subsequent callbacks are denoted with the same period as their chain.
We characterize $\Gamma_\mathcal{C}$ as follows:

\centerline{$\Gamma_\mathcal{C} := \big ( E_\mathcal{C}, T_\mathcal{C}, D_\mathcal{C}, \pi_\mathcal{C}\big )$}

\begin{itemize}
    \item \textbf{$E_\mathcal{C}$}: The cumulative WCET of an instance of the chain $\Gamma_\mathcal{C}$, i.e., $E_\mathcal{C} = \sum_{j=1}^{\|\Gamma_\mathcal{C}\|} E_{ \langle \mathcal{C},j \rangle}$.    
    \item \textbf{$T_\mathcal{C}$}: The period of an instance of $\Gamma_\mathcal{C}$.
    \item \textbf{$D_\mathcal{C}$}: The relative deadline of an instance of $\Gamma_\mathcal{C}$.

    \item \textbf{$\pi_\mathcal{C}$}: The priority of $\Gamma_\mathcal{C}$ (smaller values mean lower priority). Following the criticality-as-priority assignment~\cite{de2009scheduling}, we assume chains with higher criticality have higher priority.
\end{itemize}
We denote the $i^{th}$ instance of a chain $\Gamma_\mathcal{C}$ and its callbacks as $\Gamma_\mathcal{C}^i = \{\tau_{ \langle \mathcal{C},1 \rangle}^i, \tau_{ \langle \mathcal{C},2 \rangle}^i, \tau_{ \langle \mathcal{C},3 \rangle}^i, ... , \tau_{ \langle \mathcal{C}, \|\Gamma_\mathcal{C}\| \rangle}^i\}$. 

 The response time of the chain $\Gamma_\mathcal{C}$ is denoted as $R_\mathcal{C}$, which means the time from when the first callback $\tau_{\langle \mathcal{C},1 \rangle}$ is released until the last callback $\tau_{ \langle \mathcal{C}, \|\Gamma_\mathcal{C}\| \rangle}$ finishes execution. The chain is said to be \textit{schedulable} if $R_\mathcal{C}\le D_\mathcal{C}$. By this model, a burst release of chain instances is possible depending on the period of the chain, i.e., the next chain instance can be released before the completion of the previous instance. In other words, the response time and the deadline of a chain, $R_\mathcal{C}$ and $D_\mathcal{C}$ respectively, can be greater than its period, $T_\mathcal{C}$. Later, we will first analyze the case for constrained deadlines ($\forall \mathcal{C}: D_\mathcal{C}\le T_\mathcal{C}$) and then extend it to arbitrary deadlines ($\exists \mathcal{C}: D_\mathcal{C} > T_\mathcal{C}$).

\smallskip
\noindent\textbf{Executors.}
We consider a multi-threaded executor $\Pi$ consisting of $m$ worker threads, i.e., $\Pi = \{r_1, r_2, ..., r_m\}$. Each thread $r_k$ runs on a different CPU core to maximize concurrency and has a resource reservation for guaranteed resource supply. Hence, each thread is characterized by $r_k = (C_k^r,T_k^r)$ where $1 \leq k \leq m$, meaning that $r_k$ provides $C_k^r$ units of CPU time every $T_k^r$ units.
The supply bound function of $r_k$, $sbf_k^*(\Delta)$, which lower-bounds the amount of resource supply during an interval $\Delta$, is given by \cite{shin08compos}:
\begin{equation}\label{eq:sbf}
\begin{split}
sbf_k^*(\Delta)\!=\!\left\{\!\!\!\!
\resizebox{.85\linewidth}{!}{
    $\begin{array}{ll}
    \Delta-(\kappa+1)(T_k^r-C_k^r) & \text{if } \Delta \in [(\kappa\!+\!1)T_k^r\!-\!2\!\cdot\!C_k^r, (\kappa\!+\!1)T_k^r\!-\!C_k^r]\\
    \\
    (\kappa-1) \cdot C_k^r & \text{otherwise}\\
    \end{array}$
}
\right.
\end{split}
\end{equation}
where $\kappa=\max\bigg (\big [\frac{\Delta-(T_k^r-C_k^r)}{T_k^r}\big ], 1 \bigg )$. This captures the longest initial delay that the periodic resource reservation can incur, i.e., $2T_k^r-2C_k^r$.

For ease of integration with schedulability analysis, a linear approximation of the supply bound function has been widely used \cite{shin08compos}:
\begin{equation}\label{eq:sbf2}
\begin{split}
sbf_k(\Delta)\!=\!\left\{\!\!\!\!
\resizebox{.6\linewidth}{!}{
    $\begin{array}{ll}
    \frac{C_k^r}{T_k^r}(\Delta-2(T_k^r-C_k^r)) & \text{if } \Delta \geq 2(T_k^r-C_k^r)\\
    \\
    0 & \text{otherwise}\\
    \end{array}$
}
\right.    
\end{split}
\end{equation}
Based on this, we can derive the following.

\begin{definition}[$sbf$] The supply bound function of a multi-threaded executor $\Pi$ is given by 
\begin{equation}\label{eq:sbf_executor}
sbf_\Pi(\Delta)= \sum_{r_k \in \Pi}sbf_k(\Delta)
\end{equation}
\end{definition}
Note that worker threads of an executor do not need to be released at the same time. As long as all threads have started before $t=0$, Eq.~\eqref{eq:sbf_executor} will hold.

To find the minimum time interval required to obtain a certain amount of resource supply $x$, we use the pseudo-inverse function of $sbf_k(\Delta)$.

\begin{definition}[$\overline{sbf}_k$~\cite{guan2014general}] The pseudo-inverse function of $sbf_k$ is defined as follows:
\begin{equation}\label{eq:sbf_inverse}
\overline{sbf}_k(x)= \min \{\Delta | sbf_k(\Delta) = x \}
\end{equation}
where $x$ is the amount of resource that is needed.
\end{definition}

\section{Challenges in Multi-threaded Executor RTA} \label{challenges}

As explained earlier, a ROS 2 executor has a ReadySet which is updated at a polling point (PP), and the time interval between two consecutive polling points is called a processing window (PW). The latest work~\cite{Tang_RTSS20} derived the following lemmas on PWs in a single-threaded executor to improve analysis accuracy over \cite{Casini_ECRTS19}.

\begin{itemize}
    \item (Lemma 1 in \cite{Tang_RTSS20}) ``At most one instance of a regular callback executes in a processing window.''
    
    \item (Lemma 2 in \cite{Tang_RTSS20}) ``Let $\Gamma_\mathcal{C}$ be an arbitrary chain. Suppose $\tau_{\langle \mathcal{C}, 1\rangle}^i$ (the first regular callback of the $i^{th}$ instance of the $\Gamma_\mathcal{C}$) executes in processing window $pw_n$,
    then the earliest processing window for $\tau_{\langle \mathcal{C}, 1\rangle}^l ~ (l > i)$ to execute is $pw_{n+l-i}$."    
    \item (Lemma 3 in \cite{Tang_RTSS20}) ``At most one regular callback instance of a chain instance executes in a processing window.''    
    \item (Lemma 4 in \cite{Tang_RTSS20}) ``The regular callback instances of a chain instance execute in consecutive processing windows one by one.''
\end{itemize}



The definitions of PP and PW are rather straightforward in a single-threaded executor as there is only one thread updating ReadySet. However, in a multi-threaded executor (implemented in the \texttt{rclcpp} package of ROS 2 Galactic and newer versions to date), the ReadySet is shared among all threads.
Hence, one or more threads might become idle (i.e., ReadySet is empty or has no eligible callback to execute) and update ReadySet, while other threads are still executing their callbacks from the previous version of ReadySet. Those who were executing callbacks when ReadySet was updated cannot even notice such an update.   
Therefore, we need to revise the definitions of PP and PW for multi-threaded executors:


\begin{itemize}
    \item \textbf{Polling Point (PP)}: A time point when \textit{at least one thread} becomes idle.
    \item \textbf{Processing Window (PW)}: The time interval between two consecutive PPs, \textit{regardless of which thread triggered the new PP}. In the $n^{th}$ PW, $pw_n$, there might be some threads that are still executing callbacks whose start times were in previous PWs, $pw_{n-p}$. These callbacks are considered {\it carry-in} callbacks for $pw_n$. Also, some callbacks that started  in $pw_n$ may continue to execute in $pw_{n+q}$. Such callbacks are considered as {\it carry-out} callbacks for $pw_{n}$.
\end{itemize}

Based on the above definitions, some of the lemmas derived in \cite{Tang_RTSS20} are invalid or conditionally valid for a multi-threaded executor. Below we fix them or confirm their validity.
\begin{itemize}
    \item (Lemma 1 in \cite{Tang_RTSS20}) This lemma stays valid in a multi-threaded executor {\it only if} all chains in the system $\Gamma$ have \textit{constrained deadlines}. When chains in $\Gamma$ have \textit{arbitrary deadlines}, then at most $m$ instances of each regular callback can be in execution in the same PW, where $m$ is the number of threads in a multi-threaded executor. This is because a new PP can be triggered by any idling thread; hence, there can be $m-1$ instances carried-in from previous PWs by $m-1$ threads and 1 new instance by 1 thread that triggered the current PW.
    
    \item (Lemma 2 in \cite{Tang_RTSS20}) Since this lemma  directly follows the lemma 1 in \cite{Tang_RTSS20}, similarly, it stays valid {\it only if} all chains in $\Gamma$ have \textit{constrained deadlines}. However, for chains with \textit{arbitrary deadlines}, the earliest PW for $\tau_{\langle \mathcal{C}, 1\rangle}^l ~ (l > i)$ to execute would be 
    $pw_{n+\lfloor\frac{l-i}{m} \rfloor}$.
    
    \item (Lemma 3 in \cite{Tang_RTSS20}) This lemma remains valid with the same proof as in \cite{Tang_RTSS20}.
    
    \item (Lemma 4 in \cite{Tang_RTSS20}) This lemma is {\it not} valid anymore because the new definition of PP does not guarantee that all callbacks finish their execution when the new PW begins (i.e., the execution of some callbacks may span over multiple PWs). Hence, there is no guarantee on which PW the succeeding callback instance can get into ReadySet.
\end{itemize}

These revisions may help develop a timing analysis for a multi-threaded executor, by following the approach of \cite{Tang_RTSS20} that focuses on PW analysis. However, there are still many difficulties to be solved.
At first, analyzing when a new PP happens and how long a PW takes is not as straightforward as in a single-threaded executor. Moreover, it would be hard to determine how many PWs each callback would take to complete its execution and how many PWs exist between two consecutive callbacks. The difficulty multiplies when chains with arbitrary deadlines are considered. 
Instead, in this paper, we take a different approach: analyzing the response time of a chain without having PW in mind. Our proposed analysis is built by extending the conventional non-preemptive global task scheduling analysis~\cite{lee14improvement} and taking into account semantic differences introduced by chains, callback dependencies, and the ReadySet management. This approach also allows us to incorporate priority-driven scheduling enhancements that yield a significant benefit to critical chains.   

\section{Proposed RTA Framework}\label{RTA}

This section presents our proposed response-time analysis (RTA) framework. We first review the conventional non-preemptive fixed-priority (NP-FP) global task scheduling analysis developed for non-ROS systems (Sec.~\ref{sec:NPFP}). Then, for a single ROS 2 multi-threaded executor with {\em reentrant} callback groups, we analyze the response time of a chain with a constrained deadline (Sec.~\ref{sec:PWA_CD}). Based on this, we present priority-driven scheduling enhancements and the corresponding analysis (Sec.~\ref{sec:PPWA_CD}), and extend these to chains with arbitrary deadlines (Sec.~\ref{sec:PWA_AD}). Finally, we relax our assumptions by incorporating the effects of {\em mutually-exclusive} callback groups into our analysis (Sec.~\ref{sec:callback-groups}) and discussing how to analyze the end-to-end response time of a chain that spans across multiple executors regardless of their types (Sec.~\ref{sec:end-to-end}).

\subsection{Review of NP-FP Task Scheduling}\label{sec:NPFP}

According to the NP-FP global task scheduling analysis~\cite{lee17improved}, to obtain the worst-case response time of a given non-preemptive task $\tau_i$, we need to find the latest time that a $\tau_i$'s job starts its first unit of execution. Assuming this first unit finishes after a time interval $\Delta$ from the release of the job, the response time of $\tau_i$ is $\Delta +$ $\overline{sbf}_k(E_i -1$) where $E_i$ is the WCET of $\tau_i$ and $\overline{sbf}_k(E_i -1)$ is the minimal time interval required by any processor (worker thread $r_k$) to execute $E_i -1$. To find $\Delta$, there should be enough resource supply for the system's demand. 
For at least one unit of workload of $\tau_i$ to execute in $\Delta$, the following inequality has to hold: 
\begin{equation} \label{eq:dbf_sbf}
dbf(\Delta) < sbf_\Pi(\Delta)    
\end{equation}
where $dbf(\Delta)$ and $sbf_\Pi(\Delta)$ stand for the demand-bound function and the supply-bound function, respectively. Note that $sbf_\Pi(\Delta)=m\cdot\Delta$ is the maximum resource supply in a $m$-processor system. 

    

To determine $dbf(\Delta)$, we first need to calculate the workload of other interfering tasks in $\Delta$. Consider a set of non-preemptive tasks with constrained deadlines. The workload of a task $\tau_j$ with the WCET of $E_j$, the period of $T_j$, and the deadline of $D_j$ (i.e., $\tau_j := \langle E_j, T_j, D_j\rangle$) during an arbitrary time interval $\Delta$ can be upper-bounded by the following, as proposed in \cite{bert08}:
\begin{lemma}[Workload~\cite{bert08}]\label{lm:workload}
In an arbitrary time interval $\Delta$, the jobs of $\tau_j$ with a constrained deadline can execute up to 
\begin{equation}\label{eq:workload}
W_j(\Delta, \alpha)=
\lfloor \frac{\Delta+\alpha}{T_j}\rfloor \cdot E_j \hspace{1mm} +
\min \big(E_j, \Delta + \alpha - \lfloor \frac{\Delta+\alpha}{T_j}\rfloor \cdot T_j \big)   
\end{equation}
where $\alpha$ is an extra time to capture carry-in jobs in $\Delta$, i.e., $\alpha = R_j-E_j$ if the response time $R_j$ is known, and $\alpha = D_j-E_j$ otherwise. \footnote{Compared to the conference version of this paper that appeared at RTAS’23, Lemma 1 has been revised to correctly employ the workload function given in \cite{bert08}. We repeated the experiments and reported them in Sec.~\ref{EVAL}. We did not observe significant differences as a result of this correction.}

\end{lemma}
Then blocking from lower-priority tasks needs to be considered in $dbf(\Delta)$.
In \cite{lee14improvement}, a NP-FP global scheduling analysis is provided for periodic non-preemptive tasks, which limits the number of carry-in jobs from lower-priority tasks that can block $\tau_i$ in $\Delta$ by the number of processors, $m$. 

\begin{lemma}[LeSh~\cite{lee14improvement} in \cite{lee17improved}'s presentation] \label{lm:LeSh}
The response time of a periodic non-preemptive task $\tau_i = \langle E_i, T_i\rangle$ is upper-bounded by $R_i = \Delta +$ $\overline{sbf}_k(E_i -1)$, if $dbf(\Delta) < sbf_\Pi(\Delta)$ holds for the following $dbf(\Delta)$:
\begin{equation}\label{eq:LeSh} \begin{split}
dbf(\Delta) &= \sum_{\tau_h \in hp(\tau_i)}W_h(\Delta, R_h-E_h) \\
&+ \sum_{\tau_l \in mlp(\tau_i)} \min(E_l-1, \Delta)
\end{split}
\end{equation}
where $hp(\tau_i)$ is the set of higher-priority tasks than $\tau_i$, and $mlp(\tau_i)$ is the subset of lower-priority tasks than $\tau_i$ that give the $m$ largest $\min(E_l-1, \Delta)$ values (i.e., $|mlp(\tau_i)|=m$).
\end{lemma}

The first term of $dbf(\Delta)$ in Eq.~\eqref{eq:LeSh} upper-bounds the amount of interference from higher-priority tasks until $\tau_i$ begins execution, and the second term captures the blocking time from lower-priority jobs that have started execution before $\Delta$. If $dbf(\Delta)<sbf_\Pi(\Delta)$ holds, at least one time unit is available in $[t, t+\Delta)$ for the job of a non-preemptive task $\tau_i$ released at $t$ to start its execution; hence, the response time is bounded by $R_i=\Delta+$ $\overline{sbf}_k(E_i-1)$. This equation can be solved by a fixed-point iteration, with $\Delta=1$ as a start condition.

Note that Eq.~\eqref{eq:LeSh} uses $R_h - E_h$ for $\alpha$ of the workload function given in Eq.~\eqref{eq:workload}. To avoid the need to compute the response time of high-priority tasks in advance, $R_h - E_h$ can be replaced with $D_h - E_h$, where $D_h$ is the deadline of $\tau_h$, if $\tau_h$ is schedulable ($\because D_h \ge R_h$).

\subsection{Response Time of Chains}\label{sec:PWA_CD}

Since the ROS 2 multi-threaded executor follows global work-conserving non-preemptive scheduling,
we can derive a response-time analysis for a chain $\Gamma_\mathcal{C}$ with a constrained deadline in a similar form as in Lemma~\ref{lm:LeSh}, by finding the longest time interval $\Delta$ such that the last callback of $\Gamma_\mathcal{C}$ ($\tau_{ \langle \mathcal{C}, \|\Gamma_\mathcal{C}\| \rangle}$) has at least one unit of execution during $[t, t+\Delta)$ where $t$ is the release time of $\Gamma_\mathcal{C}$. Then, the response time of the chain could be bounded by $\Delta+$ $\overline{sbf}_k(E_{\langle \mathcal{C}, \|\Gamma_\mathcal{C}\| \rangle}-1)$, analogous to the case for NP-FP task scheduling.

However, there are several differences to consider. The first issue is an additional blocking caused by the precedence dependencies between callbacks of the chain under analysis. 
This happens because the next callback cannot start execution until its previous callback is completed, even if there exist other idle threads in the executor. While these idle threads can be utilized by callbacks from other chains at runtime, for analysis purposes, such idle threads can be assumed to be occupied by an artificial workload that needs to be added to $dbf(\Delta)$.  

\begin{figure}[t]
\centerline{\includegraphics[width=1\linewidth]{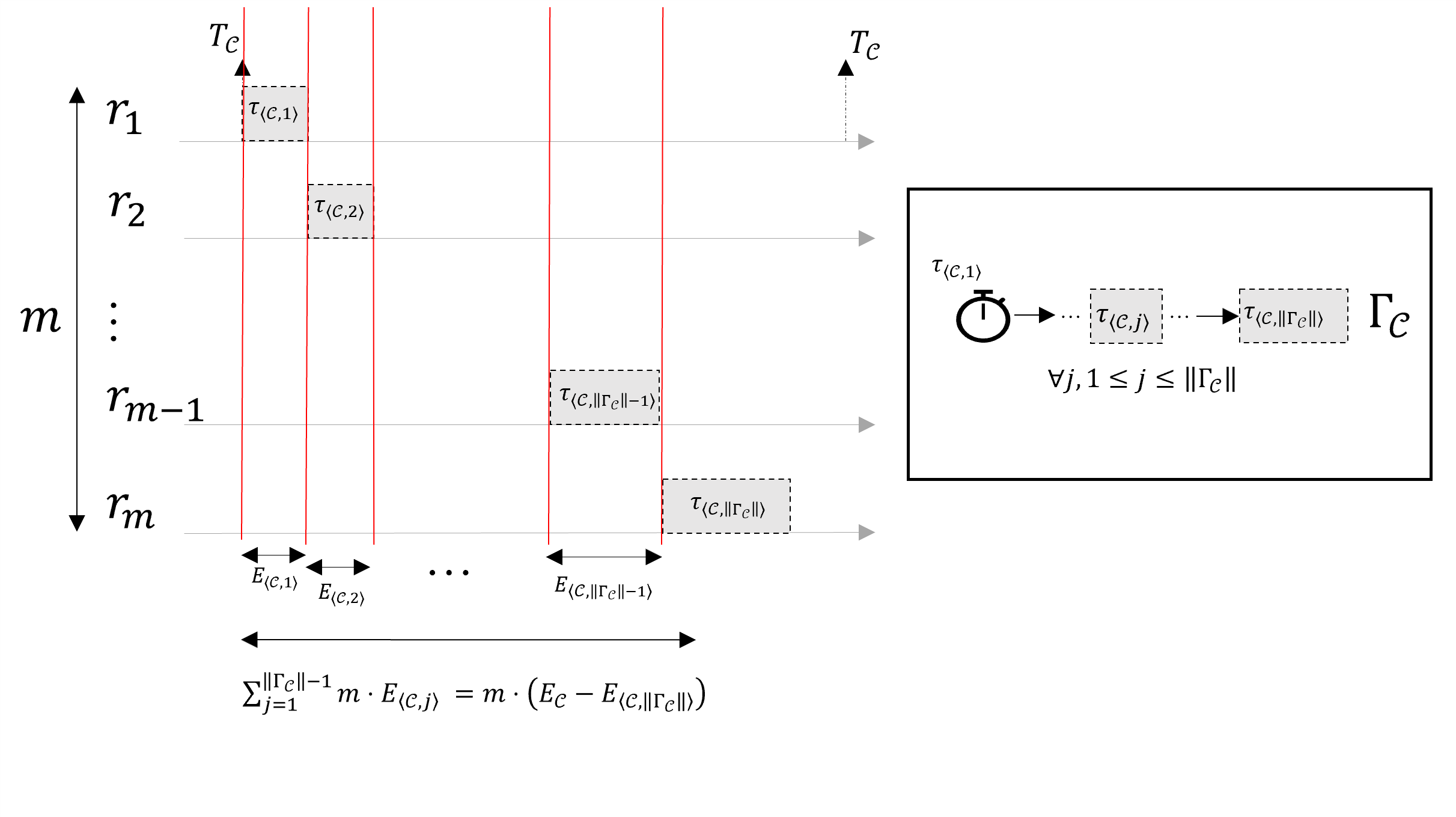}}\vspace{-10pt}
\caption{Proof of Lemma~\ref{lm:block}}
\label{fig2}
\end{figure}

\begin{lemma}\label{lm:block}
Consider two adjacent callbacks of a chain $\Gamma_\mathcal{C}$, $\tau_{\langle \mathcal{C},j \rangle}$ and $\tau_{\langle \mathcal{C},j +1 \rangle}$, on a $m$-threaded executor. The precedence-dependency blocking caused by $\tau_{\langle \mathcal{C},j \rangle}$ introduces an additional workload as interference to the start of $\tau_{\langle \mathcal{C},j+1 \rangle}$, which is upper-bounded by
\begin{equation} \label{eq:block}
    B_{\langle \mathcal{C}, j+1 \rangle} = m \cdot E_{\langle \mathcal{C}, j \rangle}
\end{equation}
%
%
\end{lemma}
\begin{proof}
    Suppose $\Gamma_\mathcal{C}$ is the only chain in the system and we are interested in the latest time that $\tau_{ \langle \mathcal{C},2 \rangle}$ can start its execution from the release time of $\Gamma_\mathcal{C}$. Since $\tau_{ \langle \mathcal{C},2 \rangle}$ cannot start until $\tau_{\langle \mathcal{C},1 \rangle}$ completes, from $\tau_{\langle \mathcal{C},2 \rangle}$'s view, 
    $\tau_{\langle \mathcal{C},1 \rangle}$ behaves as if it occupied all $m$ threads. This artificial workload equals to $m \cdot E_{\langle \mathcal{C},1 \rangle}$, which needs to be added to $dbf$ of $\tau_{\langle \mathcal{C},2 \rangle}$. By induction, as shown in Fig.~\ref{fig2}, the same happens for any callback $\tau_{\langle \mathcal{C},j \rangle}$ regardless of on which thread it is running. 
\end{proof}
Based on this, we can derive the following.
\begin{lemma}\label{lm:chain_block}
For the last callback of a chain $\Gamma_\mathcal{C}$, 
the precedence-dependency blocking caused by all of its preceding callbacks is given by $m \cdot (E_\mathcal{C}- E_{\langle \mathcal{C}, \|\Gamma_\mathcal{C}\| \rangle})$.
\end{lemma}

\begin{proof}
   By Lemma~\ref{lm:block}, $\sum_{j=1}^{\|\Gamma_\mathcal{C}\|-1} m\cdot E_{ \langle \mathcal{C},j \rangle} = m \cdot (E_\mathcal{C}- E_{\langle \mathcal{C}, \|\Gamma_\mathcal{C}\| \rangle})$.
\end{proof}

The next issue is due to the ReadySet management. 
Since ReadySet is a cached set of ready callbacks and is updated only at PPs, 
the callbacks of the chain $\Gamma_\mathcal{C}$ under analysis can be blocked multiple times by lower-priority callbacks of other chains during $\Gamma_\mathcal{C}$'s execution,
and only the execution of the last callback of $\Gamma_\mathcal{C}$ is not interfered due to the nature of non-preemptive scheduling~\cite{Casini_ECRTS19,Tang_RTSS20}. 
Therefore, we need to treat the callback instances of all other chains (regardless of their priorities) as {\it interfering} callbacks, as if they had higher priority than any callback $\tau_{\langle \mathcal{C}, j \rangle}\in \Gamma_\mathcal{C}$. 
For each interfering callback, the task-level workload function given in Eq.~\eqref{eq:workload} can be used directly since a periodic non-preemptive task $\tau_j$ in Eq.~\eqref{eq:workload} is equivalent to a callback in our model. The entire workload of an interfering chain $\Gamma_\mathcal{C'}$ can be upper-bounded as follows:

\begin{lemma}\label{lm:chain_workload}
In an arbitrary time interval $\Delta$, the instances of a chain $\Gamma_\mathcal{C'}$ with a constrained deadline can execute up to 
\begin{equation}\label{eq:chain_workload}
W_\mathcal{C'}(\Delta, \alpha)= \lfloor \frac{\Delta+\alpha}{T_\mathcal{C'}}\rfloor \cdot E_\mathcal{C'} \hspace{1mm} +
\min \big(E_\mathcal{C'}, \Delta + \alpha - \lfloor \frac{\Delta+\alpha}{T_\mathcal{C'}}\rfloor \cdot T_\mathcal{C'}\big)     
\end{equation}
where $\alpha$ is an extra time to capture carry-in instances of $\Gamma_\mathcal{C'}$.
\end{lemma}

\begin{proof}
Similar to Eq.~\eqref{eq:workload}, in the first term, the total number of instances of $\Gamma_\mathcal{C'}$ (including a carry-in) that contributes to the workload with its entire WCET ($E_\mathcal{C'}$) is obtained by $\lfloor \frac{\Delta+\alpha}{T_\mathcal{C'}}\rfloor$ and multiplied by $E_\mathcal{C'}$. The workload of the carry-out instance is bounded by the second term.    
\end{proof}

Based on Lemmas~\ref{lm:block} and \ref{lm:chain_workload}, we can obtain the following theorem to find the response time of a chain $\Gamma_\mathcal{C}$.


\begin{theorem}\label{tm:PWA_CD}
The response time of a chain $\Gamma_\mathcal{C} = \{\tau_{\langle \mathcal{C},1 \rangle}, \tau_{ \langle \mathcal{C},2 \rangle}, ..., \tau_{ \langle \mathcal{C}, \|\Gamma_\mathcal{C}\| \rangle}\}$ with a \textbf{constrained deadline} on a \textbf{standard} ROS 2 multi-threaded executor with $m$ threads is upper bounded by $R_\mathcal{C}=\Delta+$ $\overline{sbf}_k(E_{\langle \mathcal{C}, \|\Gamma_\mathcal{C}\| \rangle}-1)$, if $dbf(\Delta)<sbf_\Pi(\Delta)$ holds for the following $dbf(\Delta)$:
\begin{equation}\begin{split}\label{eq:PWA_CD}
dbf(\Delta) = m \cdot (E_\mathcal{C}- E_{\langle \mathcal{C}, \|\Gamma_\mathcal{C}\| \rangle})  +\\
\sum_{\forall \Gamma_x \in \Gamma -\{\Gamma_\mathcal{C}\}} W_x(\Delta, D_x - E_x) &
\end{split} 
\end{equation}
$dbf(\Delta)$ can be solved by a fixed-point iteration with an initial condition of $\Delta=1$.
\end{theorem}

\begin{proof}
Eq.~\eqref{eq:PWA_CD} captures all possible workloads in $[t,t+\Delta)$,  where $t$ is the release time of an instance of $\Gamma_\mathcal{C}$, until the last callback of $\Gamma_\mathcal{C}$ can start its first unit of execution. 
By Lemma~\ref{lm:block}, we know that the maximum workload caused by precedence dependencies is $m \cdot (E_\mathcal{C}- E_{\langle \mathcal{C}, \|\Gamma_\mathcal{C}\| \rangle})$. In addition, the maximum workload from the instances of other chains is upper-bounded by the sum of $W_x$ given by Lemma~\ref{lm:chain_workload}. As there is no other source of workloads in $\Delta$, the last callback of $\Gamma_\mathcal{C}$ can start at least one unit of its execution in $\Delta$ if $dbf(\Delta)<sbf_\Pi(\Delta)$ holds, and the response time of $\Gamma_\mathcal{C}$ is bounded by $\Delta+$ $\overline{sbf}_k(E_{\langle \mathcal{C}, \|\Gamma_\mathcal{C}\| \rangle}-1)$.
\end{proof}

\subsection{Priority-Driven Enhancements}\label{sec:PPWA_CD}

As explained earlier, the ReadySet management scheme of the standard ROS 2 multi-threaded executor requires our analysis to capture all callback instances of other chains as interfering callbacks, regardless of their priorities. This is particularly problematic for critical chains as their response times could be unnecessarily penalized by non-critical chains. However, prior work on a single-threaded executor~\cite{PICAS, rage} demonstrated that assigning callback priorities based on their respective chain priorities and scheduling ready callbacks by strictly based on their priorities can alleviate analytical pessimism and improve chain response time. Hence, we apply this approach to multi-threaded executors and discuss its implications in analysis.   


The implementation of priority-driven scheduling in a multi-threaded executor is rather straightforward. Instead of having each thread look for ready callbacks in ReadySet and letting it update only when it is empty, we can modify the executor code such that each thread updates ReadySet whenever it needs to choose a ready callback to execute. As a result, a newly-arrived high-priority callback does not need to wait for the other callbacks already fetched in ReadySet to finish their executions so that the next PP is triggered. Updating ReadySet in this manner may seem like a lot of extra overhead, but we found the frequency of ReadySet updates in the priority-driven multi-threaded executor is not much higher than in the standard one.
Sec.~\ref{sec:case_study} analyzes this overhead. 

For callback priority assignment, we directly adopt the chain-aware assignment scheme proposed by \cite{PICAS}, shown in Alg.~\ref{al:picas}. Basically, it makes sure that callbacks from higher-priority chains get higher callback priority than those from lower-priority chains. Within each chain, earlier callbacks get lower priority than later callbacks, i.e., $\pi_{\langle \mathcal{C}, j\rangle}< \pi_{\langle \mathcal{C}, j+1\rangle}$. As discussed in Sec.~\ref{sec:system_model}, we follow the criticality-as-priority assignment~\cite{de2009scheduling} for chains; thus, this assignment yields higher priorities to callbacks from critical chains. Also, since the standard ROS 2 executor interface does not have APIs to set callback priorities (callback priorities are determined implicitly by declaration order), we also adopted such APIs from \cite{PICAS} and applied them to the multi-threaded executor code base. 

\begin{algorithm}[b]
\caption{Chain-aware callback priority assignment~\cite{PICAS}}\label{al:picas}
    \begin{algorithmic}[1]
    \State \textbf{Input:} $\Gamma$:chains\\     
    $\Gamma \leftarrow$ sort in ascending order of chain priority $\pi$ \\
    $p \leftarrow 1$ \Comment{Initialize current priority}
    \ForAll {$\Gamma_\mathcal{C} \in \Gamma$}
    \ForAll {$\tau_{\langle \mathcal{C}, j\rangle} \in \Gamma_\mathcal{C}$}
    \State $\pi_{\langle \mathcal{C}, j\rangle} \leftarrow p$
    \State $p \leftarrow p + 1$
    \EndFor
    \EndFor
    \end{algorithmic}
\end{algorithm}

With the executor modifications for strictly priority-driven callback scheduling and the callback priority assignment, a chain of interest is no longer blocked multiple times by callbacks from lower-priority chains. Now we provide a response-time analysis for priority-driven multi-threaded executors.



\begin{theorem}\label{tm:PPWA_CD}
The response time of a chain $\Gamma_\mathcal{C} = \{\tau_{\langle \mathcal{C},1 \rangle}, \tau_{ \langle \mathcal{C},2 \rangle}, ..., \tau_{ \langle \mathcal{C}, \|\Gamma_\mathcal{C}\| \rangle}\}$ with a \textbf{constrained deadline} on a \textbf{priority-driven} ROS 2 multi-threaded executor with $m$ threads is upper-bounded by $R_\mathcal{C}=\Delta +$ $\overline{sbf}_k(E_{\langle \mathcal{C}, \|\Gamma_\mathcal{C}\| \rangle}-1)$, if $dbf(\Delta)<sbf_\Pi(\Delta)$ holds for the following $dbf(\Delta)$:
\begin{equation}\begin{split}\label{eq:PPWA_CD}
dbf(\Delta) & = m \cdot (E_\mathcal{C}- E_{\langle \mathcal{C}, \|\Gamma_\mathcal{C}\| \rangle})\\
+ & \sum_{\forall \Gamma_x \in \Gamma -\{\Gamma_\mathcal{C}\} \land \pi_x > \pi_\mathcal{C}} W_x(\Delta, D_x - E_x)\\
+ & \sum_{\forall \tau_l \in mlp(\tau_{\langle \mathcal{C}, 1\rangle})} \min(E_l-1, \Delta)
\end{split} 
\end{equation}
where $mlp(\tau_{\langle \mathcal{C}, 1\rangle})$ returns at most $m$ largest callbacks with lower-priority than $\tau_{\langle \mathcal{C}, 1\rangle}$ where it includes only one callback from each chain (i.e., {\small$ |mlp(\tau_{\langle \mathcal{C}, 1\rangle})|=min(m, |\{\Gamma_y| \pi_y <\pi_\mathcal{C}\}|)$}).
\end{theorem}

\begin{proof}

    In priority-driven scheduling, we need to upper-bound (i) the workload caused by precedence dependencies, (ii) the interference from higher-priority chain instances, (iii) and the blocking time from lower-priority chain instances. Item (i) remains the same as in Eq.~\eqref{eq:PWA_CD} since it is inherent to the chain under analysis, not subject to the scheduling policy used.

    
    By Alg.~\ref{al:picas}, 
    all callbacks of a higher-priority chain have higher priorities than the chain of interest. Therefore, to upper-bound (ii), the workload function in Lemma~\ref{lm:chain_workload} can be used for higher-priority chains than $\Gamma_\mathcal{C}$, i.e., $\forall \Gamma_x \in \Gamma -\{\Gamma_\mathcal{C}\} \land \pi_x > \pi_\mathcal{C}$. 

    For (iii), lower-priority callbacks can block the execution time of  higher-priority callbacks only when they started execution at least for one unit before $\Delta$. In a $m$-threaded executor, the number of lower-priority callbacks that can do so is up to $m$. Thus, to upper-bound the blocking time, we can find the $m$ largest callbacks from other chains with lower callback priority than $\tau_{\langle \mathcal{C}, 1\rangle}$, which is the lowest-priority callback of $\Gamma_\mathcal{C}$ by Alg.~\ref{al:picas}. Here, the precedence dependency within a chain limits up to one lower-priority callback from each chain can contribute to the blocking time. In other words, each chain cannot have more than one callback that has started before $\Delta$ and continues execution in $\Delta$. This number is bounded by $|\{\Gamma_y| \pi_y <\pi_\mathcal{C}\}|$ since only chains with lower chain priority than $\Gamma_\mathcal{C}$ have callbacks with lower callback priority than $\tau_{\langle \mathcal{C}, 1\rangle}$. As a result, the number of lower-priority callbacks $\tau_l$ contributing to the blocking time is bounded by the minimum of these two conditions, and the total blocking time is obtained by summing $\min(E_l-1, \Delta)$ since those callbacks already started one unit of execution before $\Delta$. 
\end{proof}

\subsection{Chains with Arbitrary Deadlines}\label{sec:PWA_AD}

This section extends our response-time analysis to arbitrary-deadline chains, where new instances of a chain can arrive earlier than the completion of previous instances. 
First, we need to revise the workload function given in Lemma~\ref{lm:chain_workload} because it is valid for constrained deadlines only. The implicit assumption made in that lemma (and also Lemma~\ref{lm:workload}) is that each chain has only at most one instance as carry-in and at most one instance as carry-out in a time interval $\Delta$. With arbitrary deadlines ($T_\mathcal{C} < D_\mathcal{C}$), a chain can have multiple of its previous or next instances as carry-in jobs or as carry-out jobs, respectively, as their executions can overlap. Therefore, the workload of an arbitrary-deadline chain should be captured solely based on its period.



\begin{lemma}\label{lm:chain_workload_AD}
In an arbitrary time interval $\Delta$, the instances of a chain $\Gamma_\mathcal{C'}$ with an arbitrary deadline can execute up to \begin{equation}\label{eq:chain_workload_AD}
W^*_\mathcal{C'}(\Delta, \alpha)= \lceil \frac{\Delta+\alpha}{T_\mathcal{C'}}\rceil \cdot E_\mathcal{C'}      
\end{equation}
where $\alpha$ is an extra time to capture carry-in instances of $\Gamma_\mathcal{C'}$.
\end{lemma}

\begin{proof}
    The worst-case workload can be upper-bounded by capturing the maximum possible arrivals. Therefore, we replace the floor function in Lemma~\ref{lm:chain_workload} with the ceiling function and the second term of Lemma~\ref{lm:chain_workload} is no longer needed.
\end{proof}


With the revised workload function, we can extend Theorem~\ref{tm:PWA_CD} to arbitrary deadlines as follows:

\begin{theorem}\label{tm:PWA_AD}
The response time of a chain $\Gamma_\mathcal{C} = \{\tau_{\langle \mathcal{C},1 \rangle}, \tau_{ \langle \mathcal{C},2 \rangle}, ..., \tau_{ \langle \mathcal{C}, \|\Gamma_\mathcal{C}\| \rangle}\}$ with an \textbf{arbitrary deadline} on a \textbf{standard} ROS 2 multi-threaded executor with $m$ threads is upper-bounded by $R_\mathcal{C}=\Delta + \overline{sbf}_k(E_{\langle \mathcal{C}, \|\Gamma_\mathcal{C}\| \rangle}-1)$, if $dbf(\Delta)<sbf_\Pi(\Delta)$ holds for the following $dbf(\Delta)$:

\begin{equation}\begin{split}\label{eq:PWA_AD}
dbf(\Delta) &= m \cdot (E_\mathcal{C}- E_{\langle \mathcal{C}, \|\Gamma_\mathcal{C}\| \rangle}) \\
&+ (\sum_{\forall \Gamma_x \in \Gamma} W^*_x(\Delta, D_x - E_x)) - E_\mathcal{C}
\end{split} 
\end{equation}
\end{theorem}

\begin{proof}
    The difference from Theorem~\ref{tm:PWA_CD} is the second and third terms, which capture the maximum workload from interfering chain instances.
    Here, interfering chain instances include {\em other} instances of $\Gamma_\mathcal{C}$ itself because they can cause self-interference to the instance being analyzed. By considering $\Gamma_\mathcal{C}$ in the summing term ($\forall \Gamma_x \in \Gamma$, instead of $\forall \Gamma_x \in \Gamma -\{\Gamma_\mathcal{C}\}$), $W^*_\mathcal{C}(\Delta, D_\mathcal{C}-E_\mathcal{C})$ gives all instances of $\Gamma_\mathcal{C}$ including the instance being analyzed. 
    Therefore, to avoid double-counting, $E_\mathcal{C}$ needs to be subtracted from $dbf(\Delta)$.
\end{proof}

Next is the extension of Theorem~\ref{tm:PPWA_CD} for priority-driven multi-threaded executor analysis.

\begin{theorem}\label{tm:PPWA_AD}
The response time of a chain $\Gamma_\mathcal{C} = \{\tau_{\langle \mathcal{C},1 \rangle}, \tau_{ \langle \mathcal{C},2 \rangle}, ..., \tau_{ \langle \mathcal{C}, \|\Gamma_\mathcal{C}\| \rangle}\}$ with an \textbf{arbitrary deadline} on a \textbf{priority-driven} ROS 2 multi-threaded executor with $m$ threads is upper-bounded by $R_\mathcal{C}=\Delta + \overline{sbf}_k(E_{\langle \mathcal{C}, \|\Gamma_\mathcal{C}\| \rangle}-1)$, if $dbf(\Delta)<sbf_\Pi(\Delta)$ holds for the following $dbf(\Delta)$:
\begin{equation}\begin{split}\label{eq:PPWA_AD}
dbf(\Delta) &= m \cdot (E_\mathcal{C}- E_{\langle \mathcal{C}, \|\Gamma_\mathcal{C}\| \rangle}) \\
&+ (\sum_{\forall \Gamma_x \in \Gamma \land \pi_x \ge \pi_\mathcal{C}} W^*_x(\Delta, D_x - E_x)) - E_\mathcal{C}\\ 
&+ (\sum_{\forall \tau_l \in mlp^*(\tau_{\langle \mathcal{C}, 1\rangle}, \Delta)} \min(E_l-1, \Delta)) 
\end{split}
\end{equation}
where $mlp^*(\tau_{\langle \mathcal{C}, 1\rangle}, \Delta)$ returns at most $m$ largest callbacks with lower priority than $\tau_{\langle \mathcal{C}, 1\rangle}$ while it includes only one callback from each instance of chains released in $\Delta$ (i.e., {\small $|mlp^*(\tau_{\langle \mathcal{C}, 1\rangle},\Delta)|=min(m, \sum_{\pi_y<\pi_\mathcal{C}} \lceil \frac{\Delta+D_y-E_y}{T_y}\rceil)$}).
\end{theorem}



\begin{proof}



As explained in the proof of Theorem~\ref{tm:PWA_AD}, interfering chain instances include other instances of $\Gamma_\mathcal{C}$ itself. Hence, the second term considers $\Gamma_\mathcal{C}$ by $\pi_x\ge \pi_\mathcal{C}$ instead of $\pi_x > \pi_\mathcal{C}$. As in Eq.~\eqref{eq:PWA_AD}, $E_\mathcal{C}$ needs to be deducted to avoid double-counting. 


For blocking time, we know that lower-priority callbacks can block higher-priority callbacks only when they started execution for at least one unit before $\Delta$, and the maximum number of such callbacks is bounded to $m$. The difference from the constrained-deadline case is that each of such callbacks contributing to the blocking time can be originated from each instance of other chains because there may exist multiple outstanding instances of the same chain during $\Delta$. This number is bounded by $\sum_{\pi_y<\pi_\mathcal{C}} \lceil \frac{\Delta+D_y-E-y}{T_y}\rceil$. Hence, the number of lower-priority callbacks $\tau_l$ contributing to the blocking time is bounded by the minimum of these two conditions, and the total blocking time can be obtained by maximizing the sum of $\min(E_l -1, \Delta)$.

\end{proof}

\subsection{Mutually-Exclusive Callback Groups }\label{sec:callback-groups}


Our analysis in the previous sections is for chains with reentrant callback groups. In this section, we study the effects of mutually-exclusive callback groups which introduce yet another type of precedence dependency in callback scheduling, adding more workload to $dbf(\Delta)$. As mentioned in Sec.~\ref{sec:architecture}, the use of a mutually-exclusive callback group limits any callback within this group not to be executed in parallel. In other words, the execution of callbacks within a mutually-exclusive group is all serialized. \footnote{Compared to the conference version of this paper that appeared at RTAS’23, Theorems 5 and 6 and Algorithms 2 and 3 in this section have been revised. We repeated the experiments and reported them in Sec.~\ref{EVAL}. We did not observe significant differences as a result of this correction.}


If the mutually-exclusive callback group option is enabled for some callbacks, Theorems~\ref{tm:PWA_CD} and \ref{tm:PWA_AD} (standard ROS 2 multi-threaded executor) need to be extended as follows.

\begin{theorem}\label{tm:group}
    With \textbf{mutually-exclusive callback groups}, the response time of a chain $\Gamma_\mathcal{C} = \{\tau_{\langle \mathcal{C},1 \rangle}, \tau_{ \langle \mathcal{C},2 \rangle}, ..., \tau_{ \langle \mathcal{C}, \|\Gamma_\mathcal{C}\| \rangle}\}$ on a \textbf{standard} ROS 2 multi-threaded executor with $m$ threads is upper-bounded by $R_\mathcal{C}=\Delta+\overline{sbf}_k(E_{\langle \mathcal{C}, \|\Gamma_\mathcal{C}\| \rangle}-1)$, if $dbf(\Delta) < sbf_\Pi(\Delta)$ holds for the following $dbf(\Delta)$:
    \begin{equation}\label{eq:group}
    \begin{split}
    dbf(\Delta) &= \textrm{\em RHS ~ of ~ Eq.~\eqref{eq:PWA_CD} ~ or ~ Eq.~\eqref{eq:PWA_AD}}\\
    &+\sum_{j=1}^{\| \Gamma_\mathcal{C} \|} m \cdot \textrm{\em groupmates\_load}(\tau_{\langle \mathcal{C}, j \rangle}, \Delta, \textit{CD})
    \end{split}
    \end{equation}
    where $\text{\em groupmates\_load}(\tau_{\langle \mathcal{C}, j \rangle}, \Delta, \textit{CD})$ calculates the total execution time of callback instances within the same mutually-exclusive group as $\tau_{\langle \mathcal{C}, j \rangle}$.
    The boolean variable $\textit{CD}$ is set to true in the constrained-deadline case, and false in the arbitrary-deadline case.
\end{theorem}
The pseudo-code of the $\textrm{groupmates\_load}(\tau_{\langle \mathcal{C}, j \rangle}, \Delta, \textit{CD})$ function is given by Alg.~\ref{al:group}.
\begin{algorithm}
\caption{$\textrm{groupmates\_load}(\tau_{\langle \mathcal{C}, j \rangle}, \Delta, \textit{CD})$}\label{al:group}
    \begin{algorithmic}[1]
        \State $retval = 0$;
        \State $group$ = get\_callback\_group($\tau_{\langle \mathcal{C}, j \rangle}$) - $\{ \tau_{\langle \mathcal{C}, j \rangle}\}$;
        
        \ForAll{$\tau_{\langle x, k \rangle} \in group$}
        \If{$x \neq \mathcal{C} \; \lor \; \textit{CD}=\textit{false}$}\label{ln:group_exclude_same_chain}
        \State $retval = retval + \lceil \frac{\Delta + D_x - E_x}{T_x} \rceil \cdot E_{\langle x, k \rangle}$
        \EndIf
        \EndFor
        \State $return~~ retval$;
    \end{algorithmic}
\end{algorithm}
\begin{proof}
  This theorem differs from Theorem~\ref{tm:PWA_CD} or Theorem~\ref{tm:PWA_AD} by only the extra workload added to the end of $dbf(\Delta)$. So, 
  we prove the necessity of the last term, $\sum_{j=1}^{\| \Gamma_\mathcal{C} \|} m \cdot \text{groupmates\_load}(\tau_{\langle \mathcal{C}, j \rangle}, \Delta, \textit{CD})$.


  Based on the characteristics of a mutually-exclusive callback group, once an instance of a callback $\tau_{\langle x, k \rangle}$ of the group $G_1$ executes, $m \cdot E_{\langle x, k \rangle}$ is the safe upper-bound on the blocking time that instance causes to its group-mate callbacks in $G_1$, as shown by Lemma~\ref{lm:block}. Since multiple instances of $\tau_{\langle x, k \rangle}$ can cause this blocking time to its group-mate callbacks, we need to consider the maximum possible number of instances of $\tau_{\langle x, k \rangle}$ in the time interval $\Delta$, which is bounded by $\lceil \frac{\Delta+D_x-E_x}{T_x} \rceil$ as shown in Lemma~\ref{lm:chain_workload_AD}. Thus, the maximum blocking to a callback $\tau_j$ from all of its group-mates can be bounded by $m \cdot \sum_{\forall \tau_{\langle x, k \rangle} \in G_1 \land \tau_{\langle x, k \rangle} \neq \tau_j} \lceil \frac{\Delta+D_x-E_x}{T_x} \rceil . E_{\langle x, k \rangle}$, where $\Sigma$ value is the return value of $\text{groupmates\_load}(\tau_{\langle \mathcal{C}, j \rangle}, \Delta, \textit{CD})$.

  The effect of mutually-exclusive callback groups on a chain $\Gamma_\mathcal{C}$ is maximized when all callbacks of $\Gamma_\mathcal{C}$ are in different groups and all of their group-mates block their execution as long as possible. Therefore, $m \cdot \text{groupmates\_load}(\tau_{\langle \mathcal{C}, j \rangle}, \Delta, \textit{CD})$ should be considered as part of $dbf(\Delta)$ for each $\tau_{\langle \mathcal{C}, j\rangle} \in \Gamma_\mathcal{C}$.

 As can be seen, $\text{groupmates\_load}(\tau_{\langle \mathcal{C}, j \rangle}, \Delta, \textit{CD})$ excludes the blocking from group-mate callbacks that are from $\Gamma_\mathcal{C}$, the chain under analysis, when $\textit{CD}$ is true which means in the constrained-deadline case (line~\ref{ln:group_exclude_same_chain} of Alg.~\ref{al:group}). This is because the blocking effect from callbacks of the same chain has been already considered as precedence-dependency blocking in the first term of Eq.~\eqref{eq:PWA_CD} and Eq.~\eqref{eq:PWA_AD} and there is no additional blocking from its other instances for the constrained-deadline case. Therefore, we avoid double-counting them in Alg.~\ref{al:group}.
\end{proof}

Theorems~\ref{tm:PPWA_CD} and \ref{tm:PPWA_AD} for priority-driven multi-threaded executors are changed as follows in the presence of mutually-exclusive groups.

\begin{theorem}\label{tm:group_p}
    With \textbf{mutually-exclusive callback groups}, the response time of a chain $\Gamma_\mathcal{C} = \{\tau_{\langle \mathcal{C},1 \rangle}, \tau_{ \langle \mathcal{C},2 \rangle}, ..., \tau_{ \langle \mathcal{C}, \|\Gamma_\mathcal{C}\| \rangle}\}$ on a \textbf{priority-driven} ROS 2 multi-threaded executor with $m$ threads is upper-bounded by $R_\mathcal{C}=\Delta+\overline{sbf}_k(E_{\langle \mathcal{C}, \|\Gamma_\mathcal{C}\| \rangle}-1)$, if $dbf(\Delta)<sbf_\Pi(\Delta)$ holds for the following $dbf(\Delta)$:
    \begin{equation}\label{eq:group_p}
    \begin{split}
    dbf(\Delta) &= \textrm{\em RHS ~ of ~ Eq.~\eqref{eq:PPWA_CD} ~ or ~ Eq.~\eqref{eq:PPWA_AD}}\\
    &+ \sum_{j=1}^{\| \Gamma_\mathcal{C} \|} m \cdot \textrm{\em hp\_groupmates\_load}(\tau_{\langle \mathcal{C}, j \rangle}, \Delta, \textit{CD})
    \end{split}
    \end{equation}
    where $\textrm{\em hp\_groupmates\_load}(\tau_{\langle \mathcal{C}, j \rangle}, \Delta, \textit{CD})$ sums up the execution time of higher-priority callbacks instances than $\tau_{\langle \mathcal{C}, j \rangle}$ in the same mutually-exclusive group as $\tau_{\langle \mathcal{C}, j \rangle}$. The boolean variable $\textit{CD}$ is set to true in the constrained-deadline case, and false in the arbitrary-deadline case.

\end{theorem}

The $\text{hp\_groupmates\_load}(\tau_{\langle \mathcal{C}, j \rangle}, \Delta, \textit{CD})$ function is given by Alg.~\ref{al:group_p}.

\begin{algorithm}[t]
\caption{$\text{hp\_groupmates\_load}(\tau_{\langle \mathcal{C}, j \rangle}, \Delta, , \textit{CD})$}\label{al:group_p}
    \begin{algorithmic}[1]
        \State $retval = 0$;
        \State $group$ = get\_callback\_group($\tau_{\langle \mathcal{C}, j \rangle}$) - $\{ \tau_{\langle \mathcal{C}, j \rangle}\}$;
        
        \ForAll{$\tau_{\langle x, k \rangle} \in group$}
        \If{$(x \neq \mathcal{C} \; \lor \; \textit{CD}=\textit{false}) \land \pi_{{\langle x, k \rangle}} > \pi_{{\langle \mathcal{C}, j \rangle}}$}
        \State $retval = retval + \lceil \frac{\Delta + D_x - E_x}{T_x} \rceil \cdot E_{\langle x, k \rangle}$
        \EndIf
        \EndFor
        \State $return~~ retval$;
    \end{algorithmic}
\end{algorithm}

\begin{proof}
    The proof is similar to that for Theorem~\ref{tm:group}. The only difference here is that higher-priority group-mate callbacks are the ones who execute first and cause the blocking time to a callback of interest. 
\end{proof}


\subsection{End-to-End Response Time across Executors}\label{sec:end-to-end}
The previous sections focused on the response time of a chain on one multi-threaded executor. However, a chain may span across multiple executors where each executor can be either single-threaded or multi-threaded.  To find the end-to-end response time of a chain spanning across multiple executors, one can utilize the Compositional Performance Analysis (CPA) approach, as discussed in \cite{Casini_ECRTS19}. For a set of executors, we can assign a reservation to each single-threaded or multi-threaded executor. Then, we find the response time of individual sub-chain that is executed by one executor. The analysis of such sub-chains can be done using our analysis in this section if executed by a multi-threaded executor, or the analysis from previous work~\cite{Casini_ECRTS19,Tang_RTSS20,PICAS} if executed by a single-threaded executor. One can also use our analysis with $m=1$ for a single-threaded executor, though it would give a more pessimistic upper bound than those single-thread analyses. Then, to find the end-to-end response time of a chain across multiple executors, we can sum up the response time of all sub-chains associated with the chain of interest plus the propagation delay between executors. 
This works because the activation of each sub-chain is triggered by the completion of the preceding sub-chain and can be delayed by the amount of propagation delay. 
It is worth noting that sub-chains on subsequent executors can have release jitters at runtime, e.g., the preceding sub-chain finishes earlier than its worst-case response time. 
One might suspect that such jitters introduce more interference to the other chains; however, our analysis already captures the effect of jitters with $\alpha$ in the workload functions (Eq.~\eqref{eq:chain_workload} and Eq.~\eqref{eq:chain_workload_AD}).

For example, consider a chain $\Gamma_\mathcal{C}$ consisting of $n$ sub-chains spanning across $n$ different executors as shown in Fig.~\ref{fig3}. Due to precedence dependencies among callbacks inside the chain, callbacks of a sub-chain $\Gamma_{\mathcal{C}n}$ can not start their execution until the last callback of the preceding sub-chain $\Gamma_{\mathcal{C}(n-1)}$ is finished. Once the sub-chain $\Gamma_{\mathcal{C}(n-1)}$ finishes, it takes $\delta_{n-1}$ (the propagation delay between executors $n-1$ and $n$) for the first callback of $\Gamma_{\mathcal{C}n}$ to get ready. This pattern applies to all sub-chains of $\Gamma_\mathcal{C}$. Therefore, $R_\mathcal{C}$ can be obtained by the sum of the response time of all sub-chains and their propagation delays, i.e., $R_\mathcal{C} = \sum_{i=1}^{n-1}(R_{\mathcal{C}i} + \delta_i) + R_{\mathcal{C}n}$.

It is possible that two or more sub-chains of the same chain have been assigned to the same executor. If these sub-chains are adjacent to each other, they can be merged into a single sub-chain and analyzed. If they are not adjacent, we can treat them as two independent chains; hence, each sub-chain is considered an interfering chain to the other sub-chain for analysis purposes.

\begin{figure}[t]
\centerline{\includegraphics[width=1\linewidth]{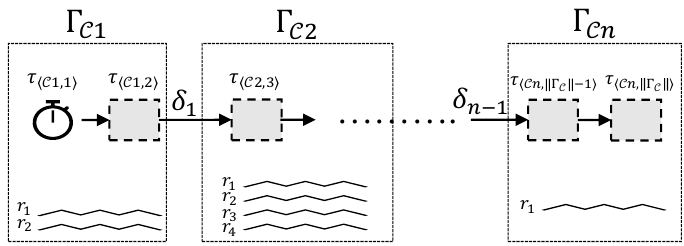}}
\caption{Chain $\Gamma_\mathcal{C}$ spanning across multiple executors}
\label{fig3}
\end{figure}





\section{Evaluation}\label{EVAL}
We evaluate our proposed work through a case study on a real platform and schedulability experiments using randomly-generated workloads. For the case study, we used the Galactic version of ROS 2 on an NVIDIA Jetson AGX Xavier (AGX) platform. Four CPU cores were used in our experiments, each set to run at the maximum frequency (2.2~GHz). 
The proposed priority-driven scheduling enhancements for multi-threaded executors were implemented as modifications to the \texttt{rclcpp} package of ROS 2.\blue{\footnote{Our source code is available at \url{https://github.com/rtenlab/ros2-picas}.}} 


\begin{figure}[t]
\centerline{\includegraphics[width=1\linewidth]{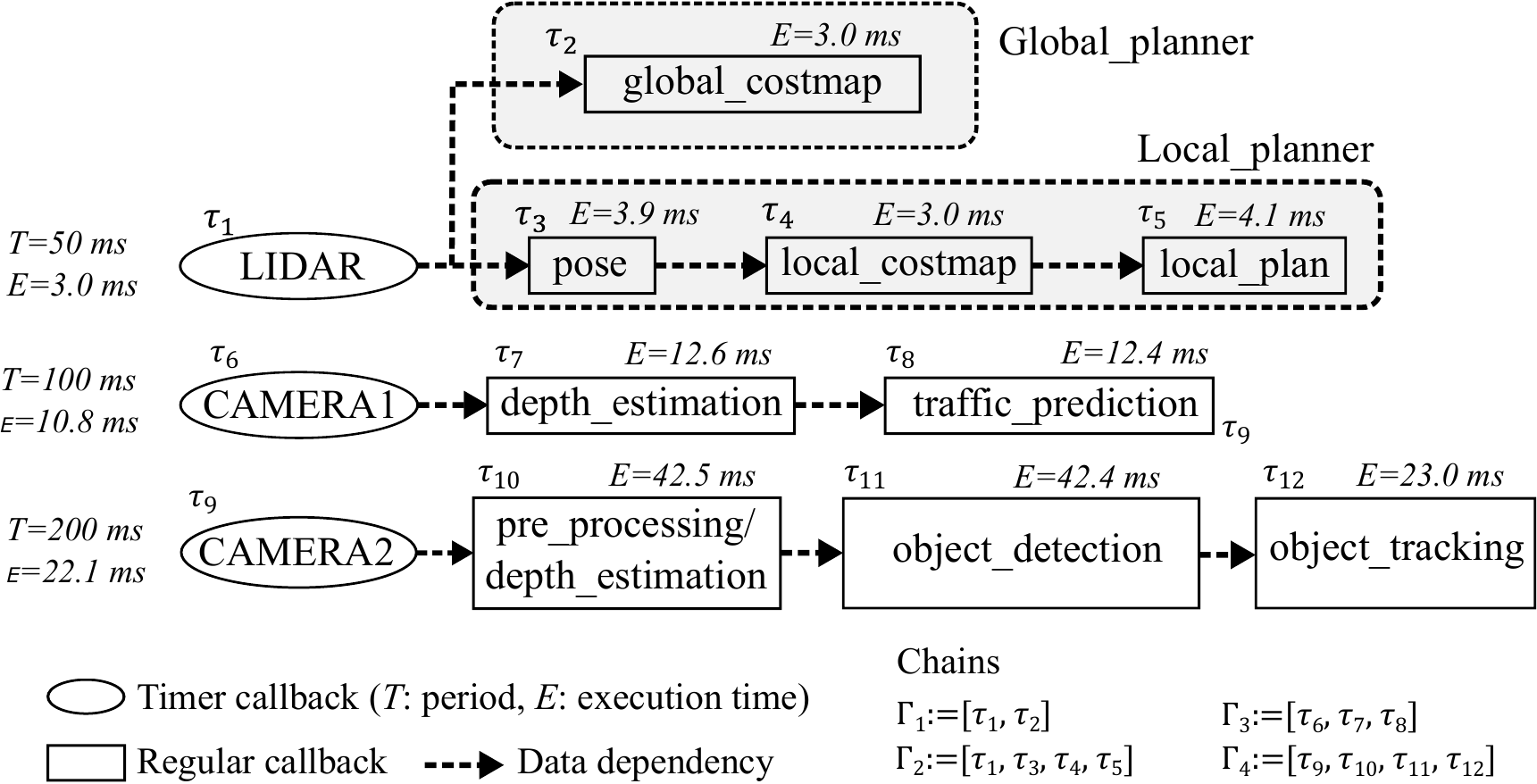}}
\caption{Chain set for case study}
\label{ex1}
\end{figure}



\begin{figure}[h]
	\centering
	\subfloat[R-D]{		\includegraphics[width=0.49\linewidth]{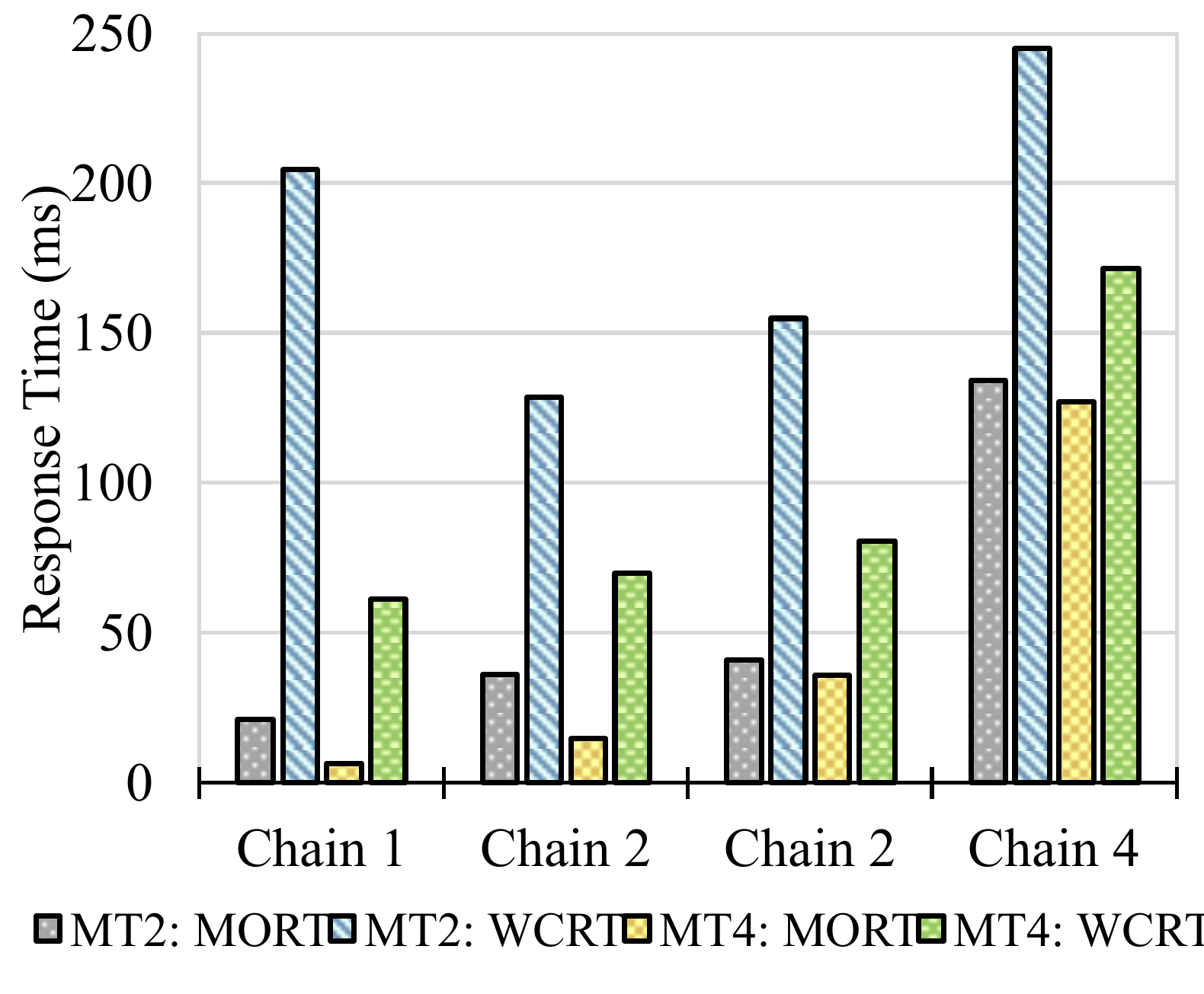}
	}
	\subfloat[R-P]{		\includegraphics[width=0.49\linewidth]{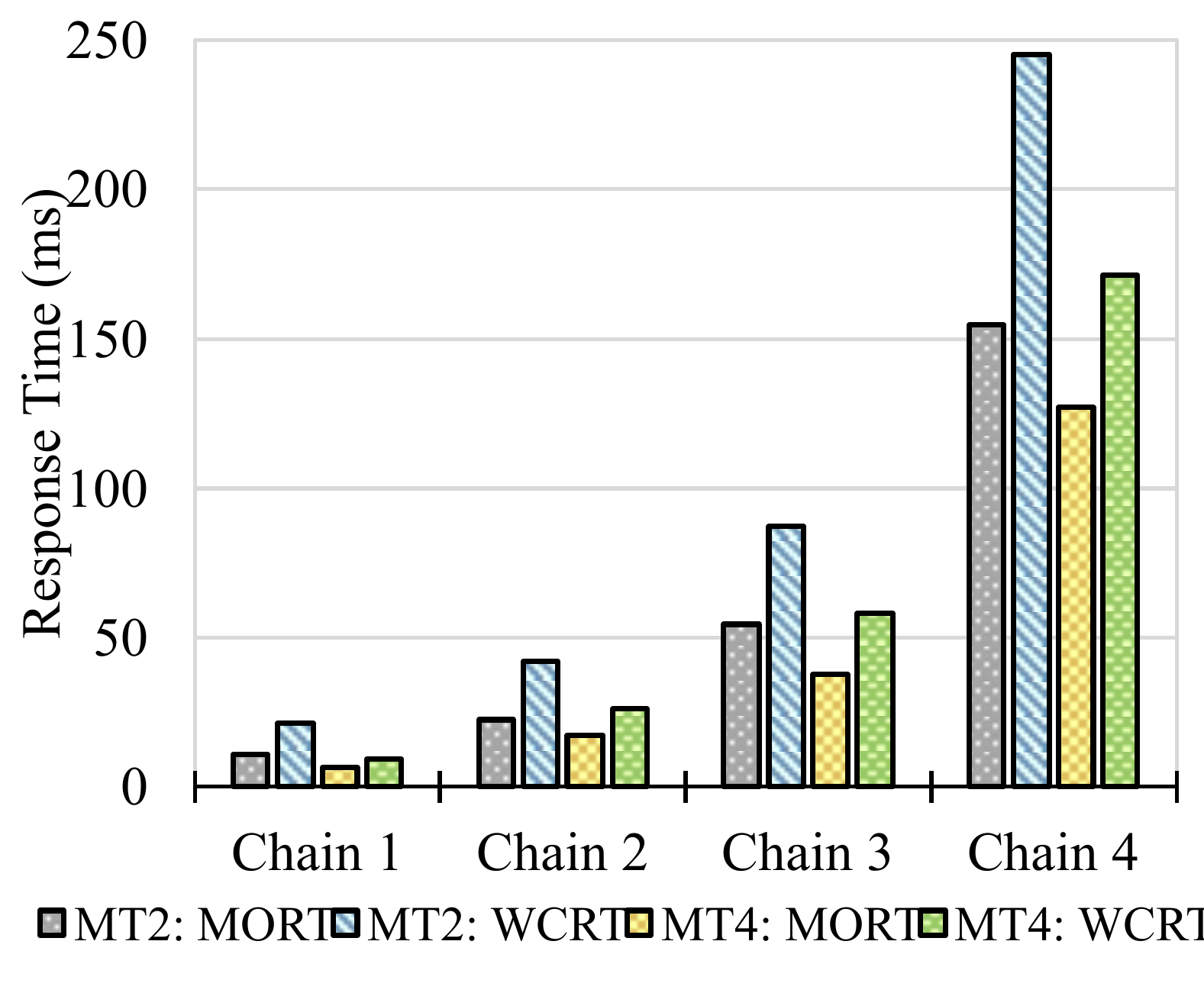}
	}\\
    \subfloat[ME-D]{
		\includegraphics[width=0.49\linewidth]{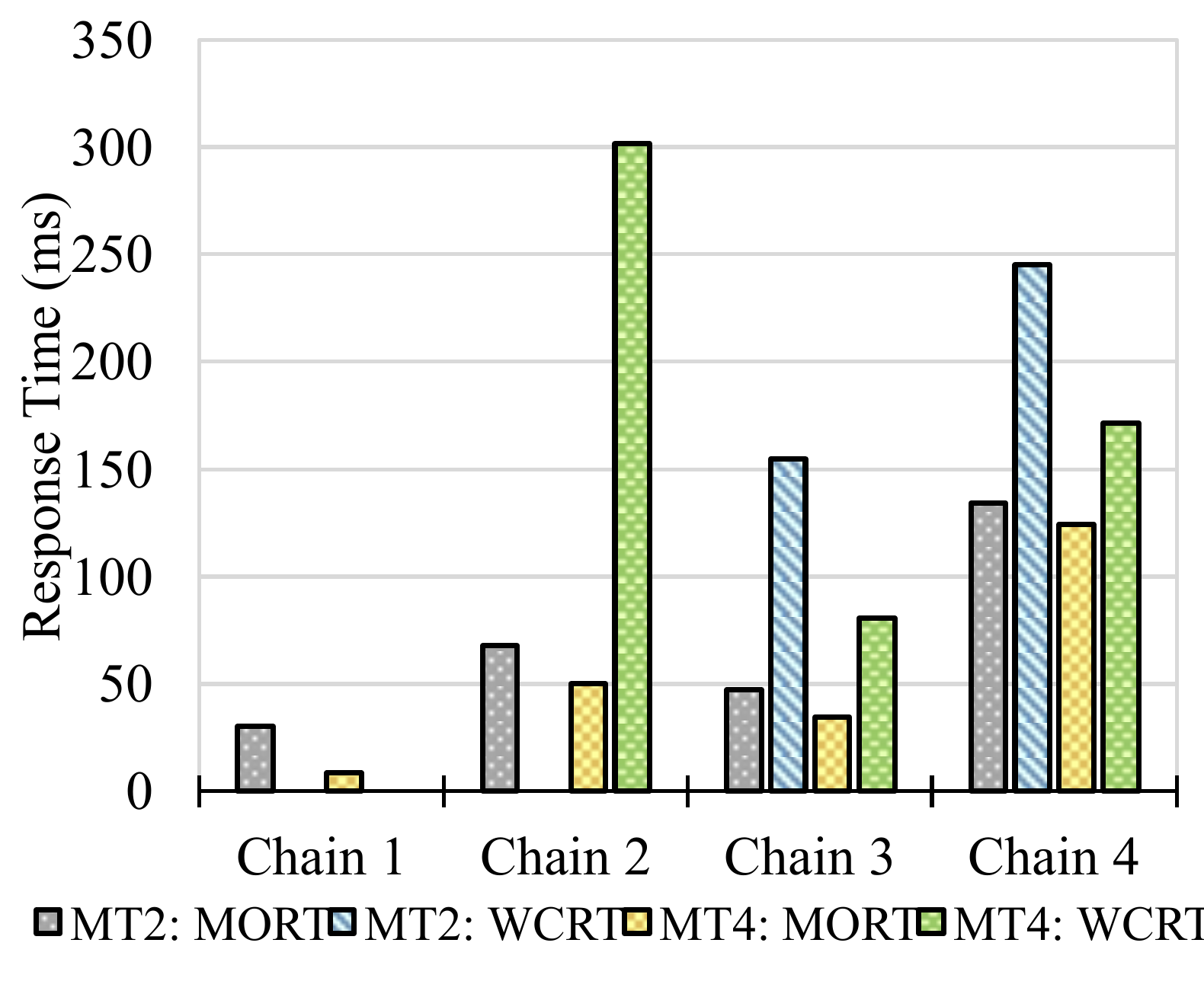}
	}
    \subfloat[ME-P]{
		\includegraphics[width=0.49\linewidth]{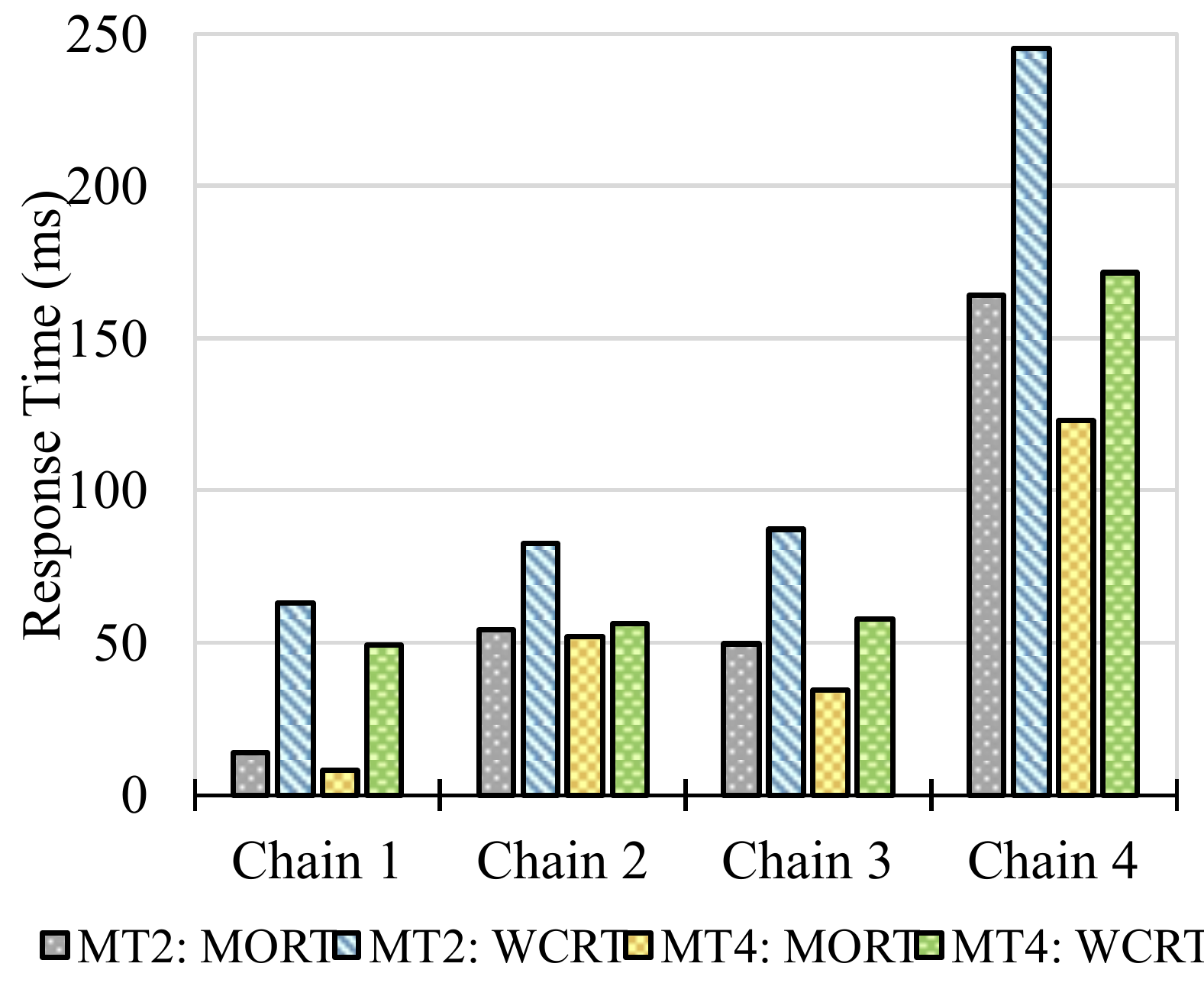}
	}
	\caption{Comparison of observed and analyzed response times}
	\label{fig:mort_wcrt}
\end{figure}

\begin{figure}[t]
	\subfloat[2*MT2-R-D]{		\includegraphics[width=0.49\linewidth]{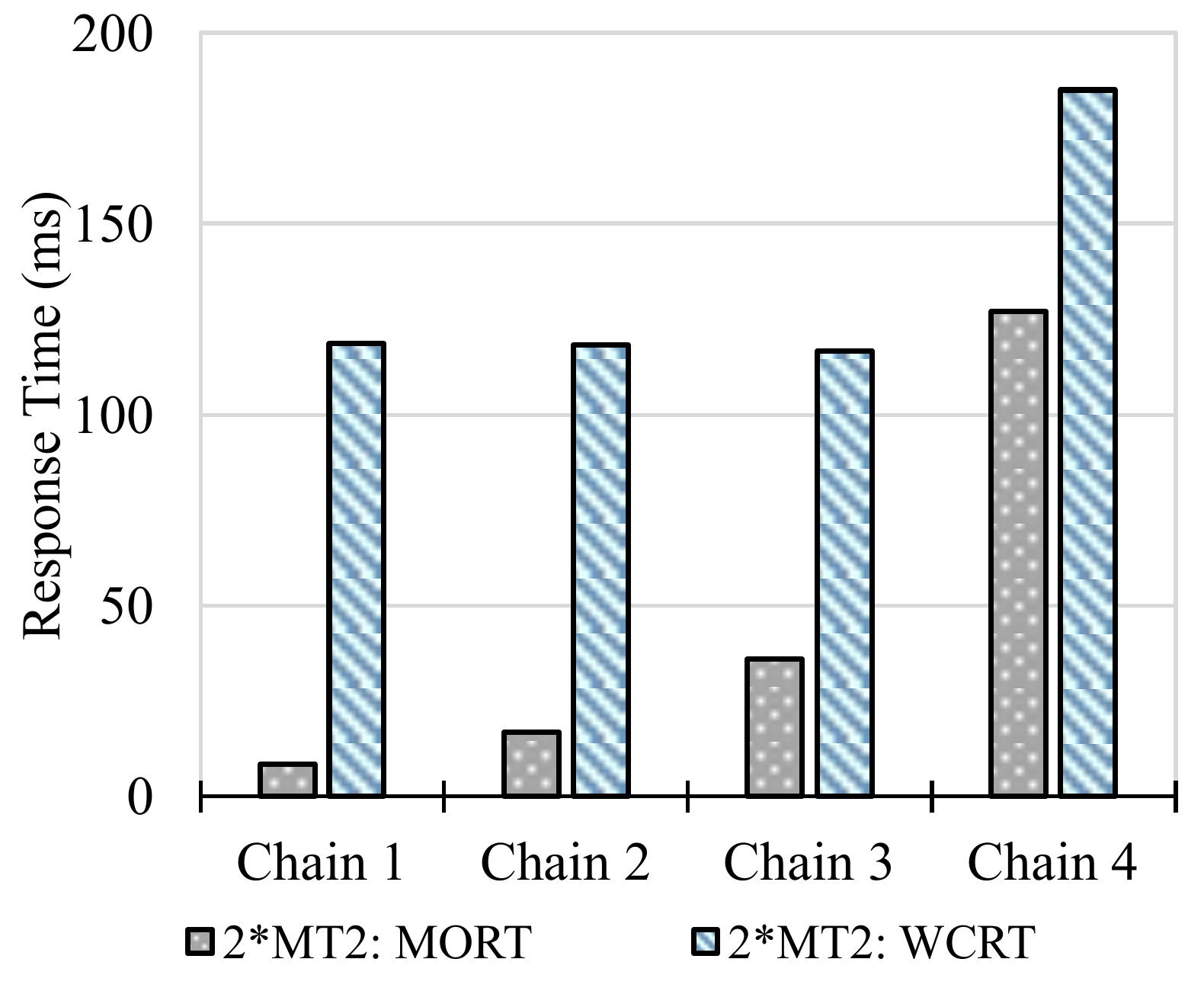}
	}
	\subfloat[2*MT2-R-P]{		\includegraphics[width=0.49\linewidth]{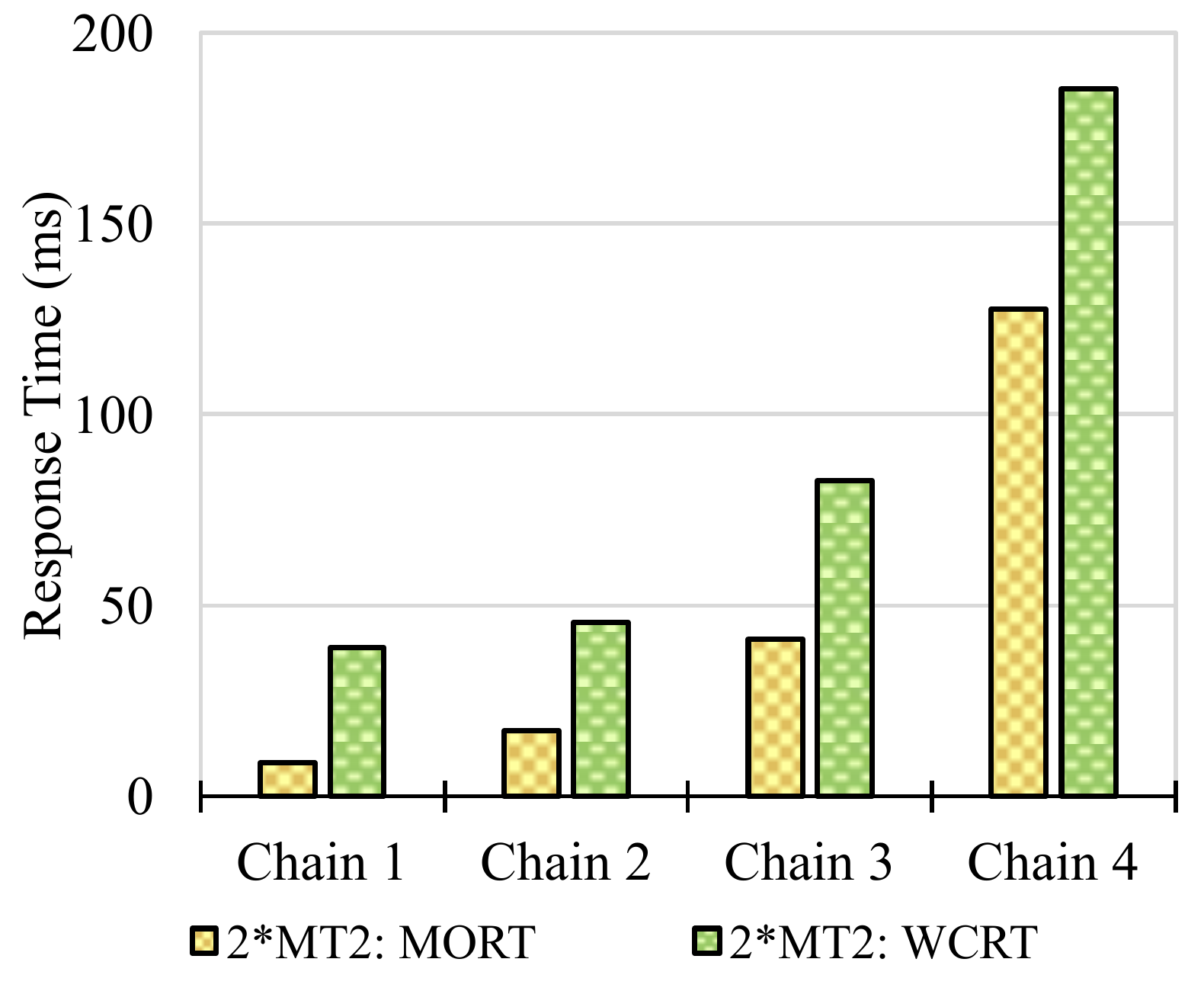}
	}    
\caption{Chains across multiple executors}
\label{fig:multi}
\end{figure}


\subsection {Case Study on AGX}
\label{sec:case_study}
The purpose of this case study is to understand the performance characteristics of the ROS 2 multi-threaded executor 
and to compare the observed response time against the worst-case bounds obtained by analysis. Fig~\ref{ex1} depicts the chain set used here, which is inspired by an autonomous robotic system. It consists of four chains: $\Gamma_1$ ($\pi_1 = 4$; highest priority), $\Gamma_2$ ($\pi_2 = 3$), $\Gamma_3$ ($\pi_3 = 2$), and $\Gamma_4$ ($\pi_4 = 1$; lowest priority). We ran this chain set for 5 minutes on AGX under each executor configuration and took the 99th percentile as the maximum observed response time (MORT) of each chain. 
For the executor configurations, we considered a multi-threaded executor with $m$ threads, where $m \in \{ 2, 4\}$.
The executor is set to use cores equal to the number of threads they have; hence, $sbf_\Pi(\Delta)=m\cdot \Delta$.
For each case, we also tested with and without our priority-driven scheduling enhancements and mutually-exclusive callback groups. We also considered cases where the chain set spans across two multi-threaded executors (each with two threads) with and without priority-driven enhancements. 
The total number of executor configurations is therefore $2\times 4 + 2 = 10$.


For analytical bounds on the worst-case response time (WCRT), we considered our analyses for constrained-deadline chains on both standard (default) and priority-driven multi-threaded executors, assuming all callbacks are in a reentrant callback group (Theorems~\ref{tm:PWA_CD} and \ref{tm:PPWA_CD}). We also examined our analyses for mutually-exclusive callback group (Theorem.~\ref{tm:group} and ~\ref{tm:group_p}) by gathering callbacks of the chains $\Gamma_1$ and $\Gamma_2$ in a mutually-exclusive group and gathering callbacks of the chains $\Gamma_3$ and $\Gamma_4$ in a reentrant group. Also, to evaluate our end-to-end latency analysis for chains spanning across multiple executors (Sec.~\ref{sec:end-to-end}), we assigned callbacks \{$\tau_1, \tau_3, \tau_6, \tau_7, \tau_9, \tau_{10}$\} to a 2-threaded executor and assigned the rest \{$\tau_2, \tau_4, \tau_5, \tau_8, \tau_{11}, \tau_{12}$\} to another 2-threaded executors. 
Since these two 2-threaded executors run on four cores of the same processor, we assumed the propagation delay between two executors is negligibly small. 


Figs.~\ref{fig:mort_wcrt} and \ref{fig:multi} compare the MORT and WCRT of each chain under different executor conditions. The caption of each sub-graph represents: the first part ``R'' and ``ME'' for the aforementioned settings of reentrant and mutually-exclusive callback groups, respectively; and, the following ``D'' and ``P'' for the ROS 2 default and the priority-driven scheduling schemes. The legend denotes: ``MT\#'' for a multi-threaded executor with \# threads (e.g., MT4 is a 4-threaded executor); ``2*MT2'' for the case of two 2-threaded executors.



\begin{figure}[t]
	\centering
	\subfloat[Chain $\Gamma_1$]{
		\includegraphics[width=0.95\linewidth]{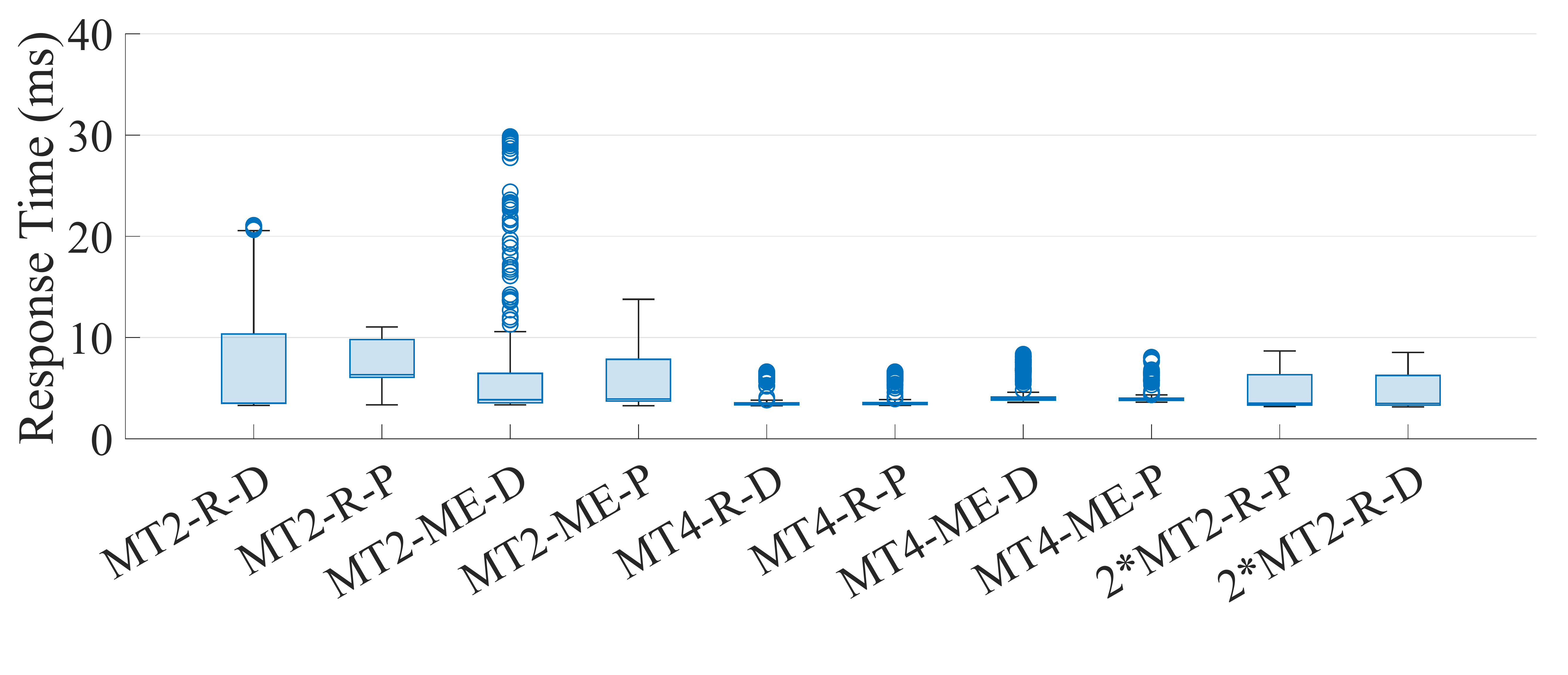}
	}\\
	\subfloat[Chain $\Gamma_2$]{
		\includegraphics[width=0.95\linewidth]{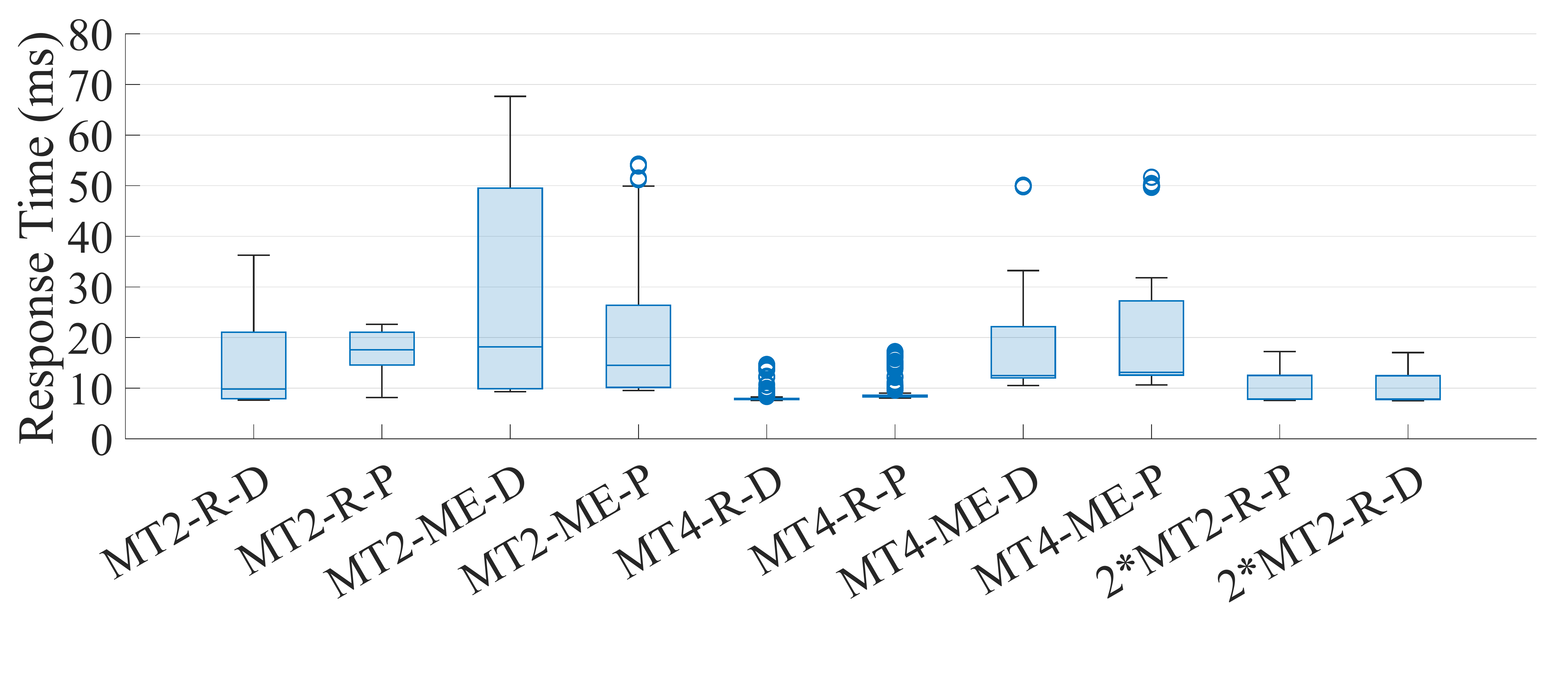}
	}\\
    \subfloat[Chain $\Gamma_3$]{
		\includegraphics[width=0.95\linewidth]{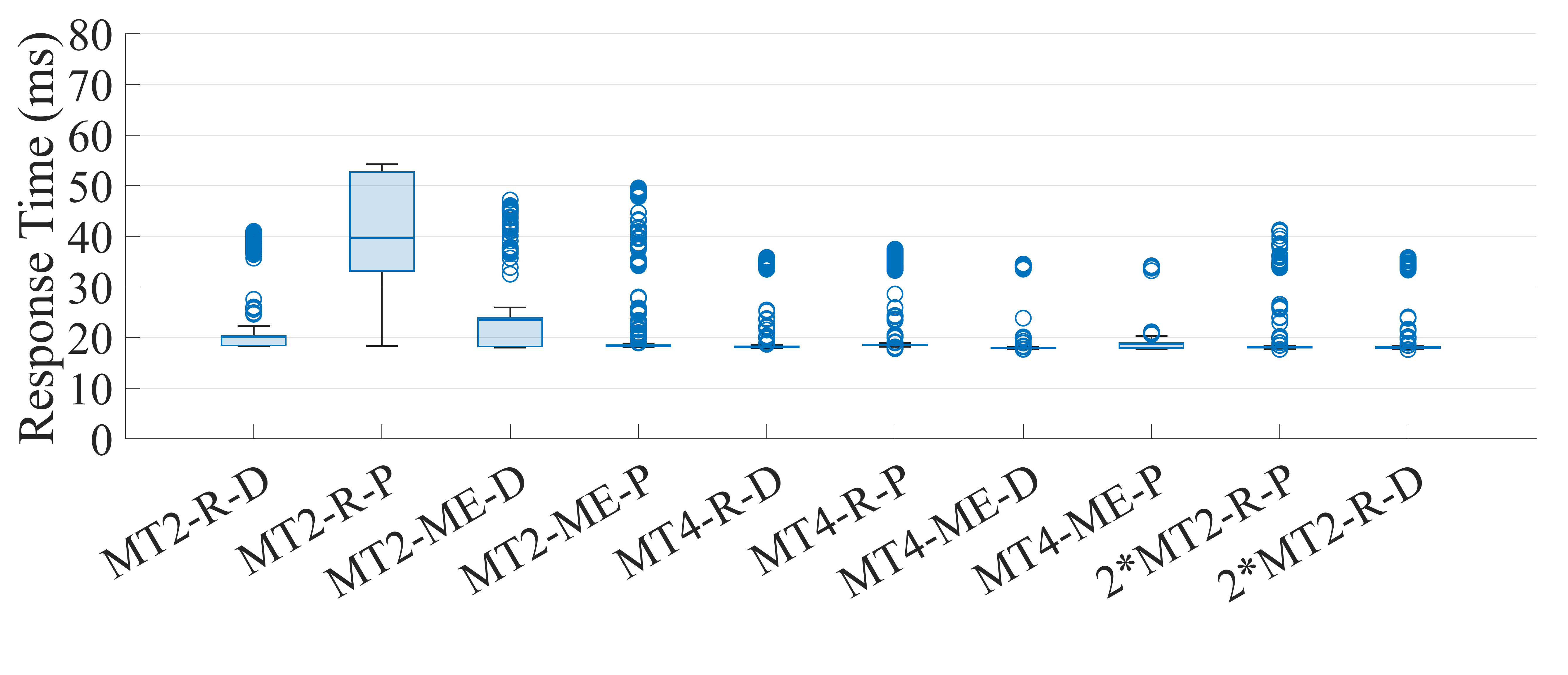}
	}\\
    \subfloat[Chain $\Gamma_4$]{
		\includegraphics[width=0.95\linewidth]{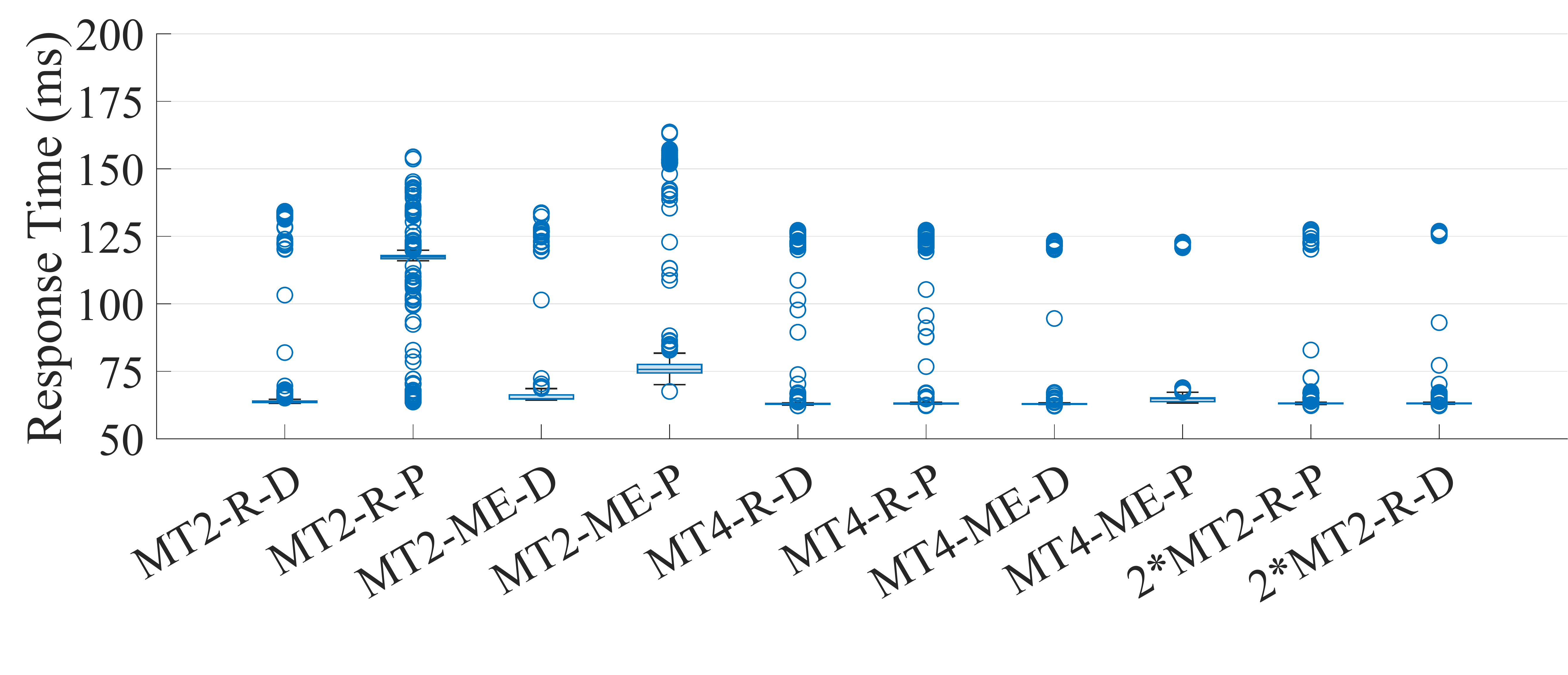}
	}
	\caption{Observed response time of chains}
	\label{fig:casestudy1}
\end{figure}


In all reentrant cases, our analysis could safely upper-bound the MORT. In the mutually-exclusive cases, when the ROS 2 default scheduler is selected, our analysis returns unbounded response time for chain 1 (with the MT2 and MT4 settings) and chain 2 (with the MT2 setting) due to the extra blocking time caused by their group-mate callbacks. For other chains using the ROS 2 default scheduler, or for all chains when the ROS 2 priority-driven scheduler is applied, our ME analysis perfectly upper-bounds the MORT. The priority-driven scheduling scheme outperforms the default ROS 2 scheduler in both reentrant and mutually-exclusive callback groups, especially for high-priority chains. These results indicate that both MORT and WCRT are reduced for high-priority chains with priority-driven scheduling and this enhancement not only reduces analysis pessimism but also improves the response time of critical chains at runtime. Although, in some cases, the response time of low-priority chains slightly increases with priority-driven scheduling, this is the cost to improve the response time of high-priority ones.


\begin{figure}[t]
\centerline{\includegraphics[width=1\linewidth]{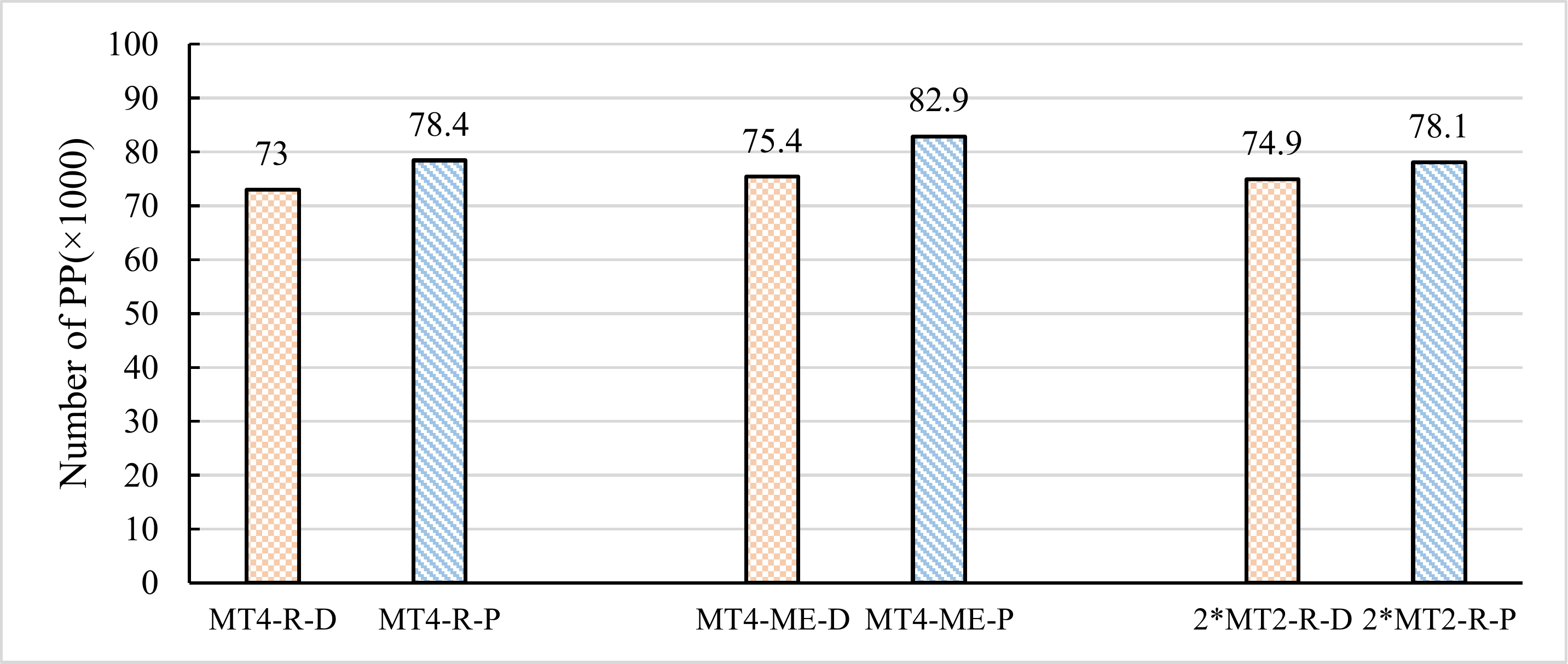}}
\caption{Overhead w.r.t. the number of polling points (PP)}
\label{fig:overhead}
\end{figure}

Fig.~\ref{fig:casestudy1} illustrates more details on the observed response time distributions of the four chains under various executor configurations. 
It is clearly shown that the use of priority-driven scheduling reduces variations in the observed response time of critical (high-priority) chains, thereby improving perceived predictability.

To assess the overhead of our priority-driven scheduling scheme, we measured the number of ReadySet updates during 5 minutes of running the case study under various executor configurations. 
Fig.~\ref{fig:overhead} illustrates the results.
As can be seen, the priority-driven scheduling scheme introduces less than 10\% of additional ReadySet updates compared to the default scheme. 

\subsection{Schedulability Experiments}
In this section, we explore the schedulability ratio of the proposed RTA framework using randomly-generated chain sets with constrained and arbitrary deadlines. 
The following abbreviations are used for each of the analysis approaches:
\begin{itemize}
\item \textbf{PWA\_CD}: Proposed Worst-case Analysis (PWA) for chains with Constrained Deadlines (CD) -- Theorem~\ref{tm:PWA_CD}
\item \textbf{PPWA\_CD}: Priority-driven PWA for CD  -- Theorem~\ref{tm:PPWA_CD}
\item \textbf{PWA\_AD}: PWA for Arbitrary Deadlines (AD) -- Theorem~\ref{tm:PWA_AD}
\item \textbf{PPWA\_AD}: Priority-driven PWA for AD -- Theorem~\ref{tm:PPWA_AD}
\end{itemize}

We used the UUniFast algorithm \cite{bini2005measuring} with a total utilization from 0.8 to 4.0 with the step of 0.4 to generate 1,000 random chain sets, where each set includes 5 chains and each chain has 10 callbacks. At first, chains were generated for the constrained-deadline case with $T_\mathcal{C}=D_\mathcal{C}$ using the UUniFast algorithm. Then, for arbitrary deadlines, we duplicated these chains and doubled their deadlines. The number of threads is set to $m=4$.

\begin{figure}[t]
\centerline{\includegraphics[width=1\linewidth]{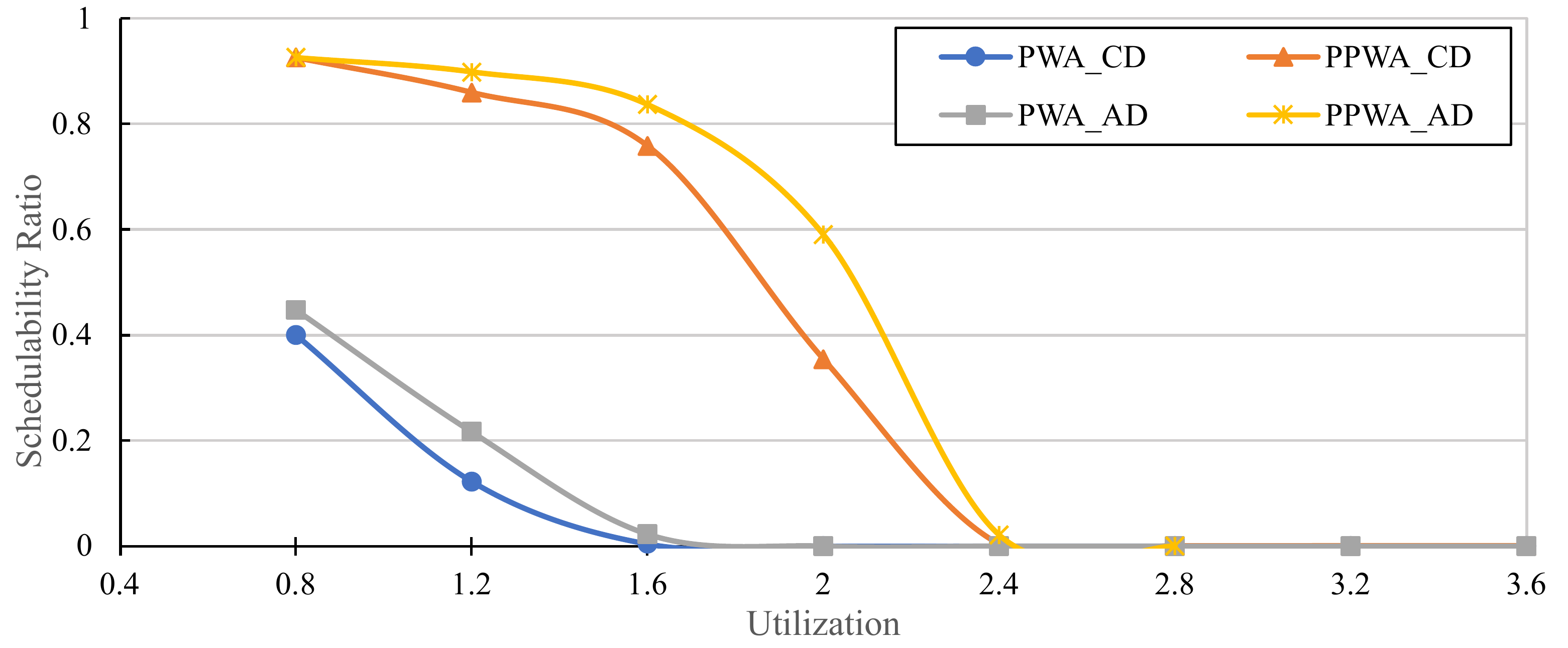}}
\caption{Varying utilization of chain sets}
\label{fig:varying_utilization}
\end{figure}

Fig.~\ref{fig:varying_utilization} shows the schedulability ratio of chain sets by our four analysis approaches as the chain-set utilization increases. 
The schedulability ratio might be seen as rather low for the given utilization. However, it is worth noting that, in chain scheduling, the precedence dependencies among callbacks cause resource waste and therefore lead to a less schedulability ratio than conventional task scheduling with no dependencies.
In general, priority-driven scheduling (PPWA\_CD and PPWA\_AD) significantly outperforms the default ROS 2 scheduling policy (PWA\_CD and PWA\_AD), with as much as 55\% point higher in schedulability. This is primarily due to the fact that high-priority chains experience less delay from lower-priority chains under PPWA since it fetches ready callbacks immediately and schedules them strictly based on their assigned priority. The results of AD cases look slightly better than their CD counterparts. However, given that AD cases are tested with chain sets with doubled deadlines, it can be said that AD analysis has much higher pessimism than CD analysis. This is somewhat expected from our analysis in Theorem~\ref{tm:PPWA_AD} which had to take into account multiple outstanding (= released and unfinished) instances of each chain.

We also explored the schedulability ratio of our analyses as the number of threads increases. Similar to above, we generated 1,000 random chain sets, each including 5 chains with 10 callbacks per chain. The utilization of each chainset is fixed at 1.0 in this experiment. The results are shown in Fig.~\ref{fig:varying_number_of_threads}. Note that scheduling a chain set with a utilization equal to 1.0 on a single thread is almost infeasible due to the precedence dependencies among callbacks. The schedulability ratio improves with a more number of threads, but the degree of improvement appears differently for each analysis method. For PPWA\_CD and PPWA\_AD, the schedulability ratio increases sharply at $m=2$ and becomes almost plateau afterward. On the other hand, the increase of PWA\_CD and PWA\_AD is slower and is almost linear to the number of threads. We can see that our proposed priority-driven enhancement can achieve better schedulability with a fewer number of threads.

\begin{figure}[t]
\centerline{\includegraphics[width=1\linewidth]{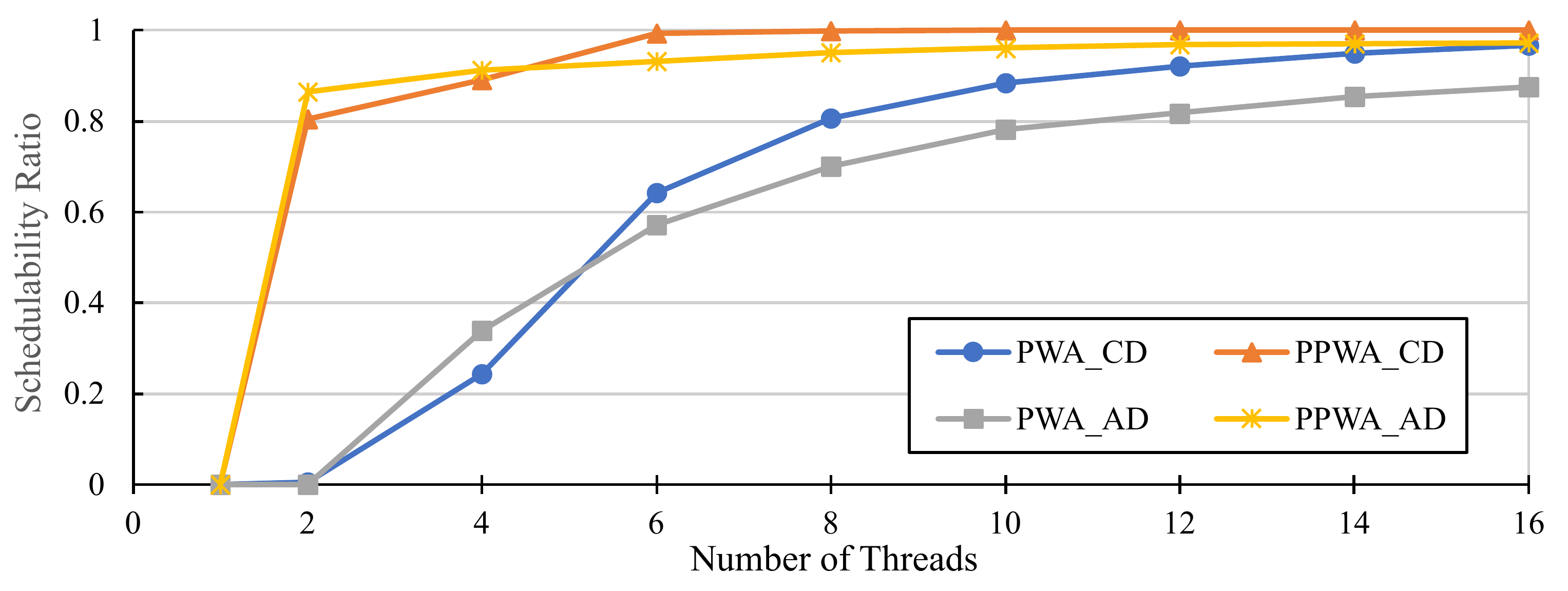}}
\caption{Varying number of threads}
\label{fig:varying_number_of_threads}
\end{figure}

\begin{figure}[t]
\centerline{\includegraphics[width=1\linewidth]{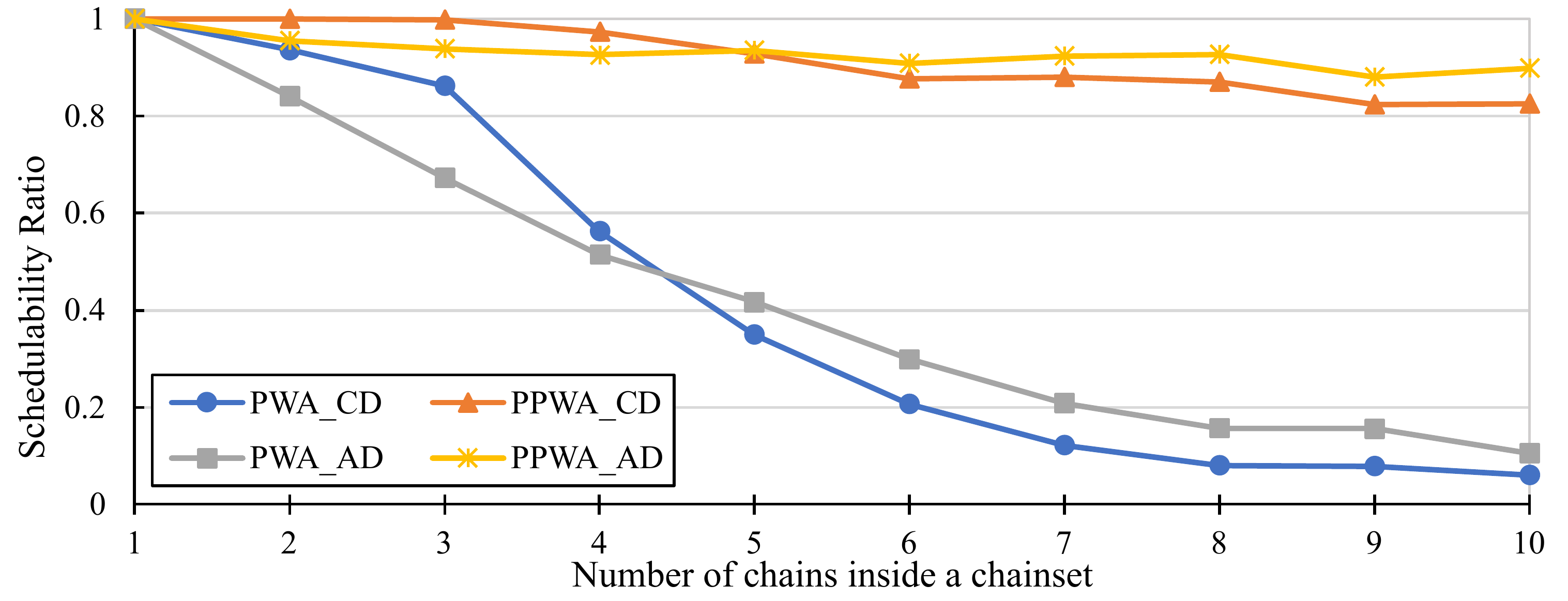}}
\caption{Varying number of chains}
\label{fig:sched_C}
\end{figure}

Lastly, we explored the impact of the number of chains on schedulability. Here, we kept the total utilization to 1 and varied only the number of chains. Each generated chain has 10 callbacks, regardless of how many chains are generated for each chain set. Fig.~\ref{fig:sched_C} depicts the results. The schedulability ratio decreases with the number of chains for all analysis methods, meaning interference increases with the number of chains although the total utilization remains the same.
Recall that, due to precedence dependencies, only one callback of each chain instance can interfere with the chain under analysis. However, by increasing the number of chains, the number of callbacks from chains increases, and also the number of chain instances in a time interval $\Delta$ could increase. However, Fig.~\ref{fig:sched_C} shows that priority-driven scheduling is less affected by this issue because it limits the source of interference to only higher-priority chains. We therefore conclude that priority-driven scheduling brings a significant benefit to real-time chains on ROS 2 multi-threaded executors.



\section{Conclusion}
In this paper, we proposed a comprehensive response-time analysis framework for chains running on ROS 2 multi-threaded executors. We analyzed the timing behavior of the default scheduling scheme of ROS 2 multi-threaded executors and proposed priority-driven scheduling as an improvement over the default one. Our framework can analyze chains with both arbitrary and constrained deadlines and can take into account the effect of mutually-exclusive callback groups. 
We conducted the evaluation with a case study on NVIDIA Jetson AGX Xavier and schedulability experiments using randomly-generated chains. The results demonstrate that our analysis framework can safely upper-bound response times under various conditions. In addition, our priority-driven scheduling for ROS 2 multi-threaded executors not only reduces the response time of critical chains but also improves analysis results. 

As ROS 2 is becoming more popular in academia and industry for more complex robotic systems with safety-critical features, more active research efforts need to be made to ensure timing correctness and improve system efficiency in ROS-based systems. There are several interesting future directions, including real-time support for accelerators, synchronization, memory-induced interference, which have been studied for conventional systems but not in the context of ROS or similar middleware environments. We believe that our work fills an important knowledge gap in this area and more interesting work could be built upon our framework. 

\section*{Acknowledgment}
This work was sponsored by the National Science Foundation (NSF) grant 1943265 and the Office of Naval Research (ONR) grant N00014-19-1-2496.


\bibliographystyle{IEEEtran}
\bibliography{references.bib}

\end{document}